\documentclass[a4paper,USenglish,cleveref, autoref, thm-restate]{lipics-v2021}

\usepackage{amssymb,amsmath,amsthm,multirow,color,xcolor,cancel,bbm,tikz,tikz-cd,hyperref,tcolorbox}
\usepackage{hyperref}

\usetikzlibrary{decorations.markings}
\usetikzlibrary{shapes.multipart}
\usetikzlibrary{decorations.pathmorphing}

\hypersetup{
    colorlinks=true,
    linkcolor=blue,
    citecolor=cyan,
    filecolor=magenta,      
    urlcolor=blue,
    pdftitle={UFP},
    }

\def\B{\{0,1\}}
\def\ONE{\mathbf{1}}
\def\ZERO{\mathbf{0}}

\def\S{\mathcal{S}}

\def\Sol{\mathrm{Sol}}
\def\PosSol{\Sol^{\textrm{\tiny $+$}}}
\def\Fixe{\mathrm{Fix}}
\def\lcm{\mathrm{lcm}}
\def\N{\mathbb{N}}
\def\Ns{\N^*} 
\def\Per{\mathrm{Per}}
\def\id{\mathrm{id}}

\newcommand{\EM}[1]{{\em\textcolor{Maroon}{#1}}}
\newcommand{\EMM}[1]{{\textcolor{Maroon}{#1}}}

\newcommand{\decisionpb}[3]{
\begin{tcolorbox}[
left=1mm,
right=1mm,
top=1mm,
bottom=1mm,
colframe=black,
colback=white,
arc=3mm,      
boxrule=0.6pt 
]
{{\sc #1}\\{\bf Input:} #2\\{\bf Question:} #3}
\end{tcolorbox}
}

\def\qsat{\textrm{\sc QSAT$_2$}}
\def\succinct{\textrm{\sc Succinct-4SAT}}
\def\UFP{\textrm{\sc UFP}}
\def\LCUFP{\textrm{\sc loop-to-cycle-UFP}}
\def\dUFP{\textrm{\sc $d$-UFP}}
\def\dLCUFP{\textrm{\sc $d$-loop-to-cycle-UFP}}

\def\OSvsD{\textrm{\sc OS-vs-D}}
\def\PeriodPiece{\textrm{\sc Period Piece}}

\def\NEXPTIME{\mathrm{\sf NEXPTIME}}
\def\Poly{\mathrm{\sf P}}
\def\NP{\mathrm{\sf NP}}
\def\coNP{\mathrm{\sf coNP}}
\def\NPoNP{\NP^{\NP}}

\usepackage[dvipsnames]{xcolor}

\DeclareFontShape{T1}{lmr}{m}{scit} { <-> ssub * cmr/m/sc }{}

\usepackage{placeins} 
\makeatletter 
\newcommand{\newpagexorapx}[0]{\ifthenelse{\equal{\axp@appendix}{inline}} {\newpage} {}}
\makeatother

\usepackage{silence}

\usepackage[appendix=append,bibliography=common]{apxproof}

\newtheoremrep{theorem}[thm]{Theorem}
\newtheoremrep{claim}[thm]{Claim}
\newtheoremrep{lemma}[thm]{Lemma}
\newtheoremrep{corollary}[thm]{Corollary}
\newtheoremrep{observation}[thm]{Observation}
\newtheoremrep{proposition}[thm]{Proposition}
\newtheoremrep{conjecture}[thm]{Conjecture}

\nolinenumbers

\title{\texorpdfstring{Period Piece: $\NEXPTIME$-complete problems\\ on digraphs via Boolean network dynamics
}{Period Piece: NEXPTIME-complete problems on digraphs via Boolean network dynamics
}} 

\titlerunning{Period Piece: $\NEXPTIME$-complete problems on digraphs via Boolean network dynamics} %

\author{Florian Bridoux}{Université Côte d’Azur, CNRS, I3S, Sophia Antipolis, France.}{bridoux@i3s.unice.fr}{https://orcid.org/0000-0002-8546-5995}{}
\author{Maximilien Gadouleau}{Department of Computer Science, Durham University, Durham, UK}{m.r.gadouleau@durham.ac.uk}{https://orcid.org/0000-0003-4701-738X}{Supported by Leverhulme Trust Research Project Grant RPG-2025-400.}
\author{David C. Kutner}{Department of Computer Science, Durham University, Durham, UK}{david.c.kutner@durham.ac.uk}{https://orcid.org/0000-0003-2979-4513}{Supported by Leverhulme Trust Research Project Grant RPG-2025-400.}
\author{Adrien Richard}{Université Côte d’Azur, CNRS, I3S, Sophia Antipolis, France.}{richard@unice.fr}{https://orcid.org/0000-0003-3360-5869}{}

\funding{Supported by ANR-24-CE48-7504 “ALARICE” project.}

\authorrunning{F. Bridoux, M. Gadouleau, D. C. Kutner and A. Richard} %

\Copyright{Florian Bridoux, Maximilien Gadouleau, David C. Kutner, Adrien Richard} 

\ccsdesc{Theory of computation~Complexity classes}
\ccsdesc{Mathematics of computing~Discrete mathematics}

\keywords{Boolean networks, NEXPTIME, Directed graph, Fixed point} %

\category{} %

\relatedversion{}
\begin{document}
\maketitle

\begin{abstract}
    The central object in this work is the \EM{Boolean network}, a function $f:\{0,1\}^V \to \{0,1\}^V$. Often identified with $f$ is its \EM{dynamics} $D_f$, the digraph with exactly $\{x \to y : f(x)=y\}$ as its arcs.
    The \EM{interaction graph} of $f$ is the digraph $G$ on $V$ which can be thought of as representing the dependencies of the components of $f$ on one another. 
    On the other hand, model-checking for relational structures in general (and digraphs in particular) is a major field of study in algorithms and complexity, concerned with problems asking whether an input graph satisfies a specified formula. 
    
    We consider problems at the intersection of these paradigms: for a given property $\phi$, we say that a digraph $G$ is $\phi$-dynamic if $G$ is the interaction graph of some Boolean network $f$ such that $D_f \models \phi$. The corresponding decision problem $\phi$-\textsc{Dynamic} asks whether $G$ is $\phi$-dynamic.
    Problems of this form are well-studied in the Boolean networks community. For example, it is a well-known result that, letting $\phi$ be ``No fixed points'' (i.e. $\nexists x : x \to x$), $\phi$-\textsc{Dynamic} is equivalent to the structural digraph problem \textsc{Non-acyclic}, which is solvable in linear time.  
    
    We establish, for the first time, properties $\phi$ for which $\phi$-\textsc{Dynamic} is computationally hard. These include the ``Unique fixed point'' property $\exists! x: x \to x$, and a range of properties concerning the periodic points of $f$ (i.e. vertices of $D_f$ which lie on a cycle).
    In so doing, we introduce, to our knowledge, the first naturally-defined languages on digraphs which are complete for $\NEXPTIME$. Typically, languages complete for this class either take as input a circuit which is a succinct representation of some other object, or include a binary representation of an integer. $\phi$-\textsc{Dynamic} does neither. 

    More widely, we exhibit an infinite family of natural properties $\phi$ such that $\phi$-\textsc{Dynamic} is $\NEXPTIME$-complete, which we call \textsc{Period Piece}. As a special case, these include all properties of the form ``Unique cycle of length $k$'' and ``Exactly $k$ fixed points'' for any $k \ge 1$. 
    On the other hand, we show that if the input digraph has bounded in-degree, then for those same problems $\phi$-\textsc{Dynamic} collapses down to $\NPoNP$, the second level of the polynomial hierarchy, though it remains $\NPoNP$-complete.  
\end{abstract}

\newpage

\noindent
\textit{Full proofs are deferred to the %
appendix. }
\medskip\\
\noindent\textbf{On the use of AI tools:}
\textit{No GenAI tools were used in the preparation of this manuscript, nor in the development of the theory presented herein. 
If the article is accepted, a commercially available model will be used to check for any missing issues or bugs prior to uploading the camera-ready version. 
This disclosure will be updated to describe this use.}

\newpage

\section{Introduction}

\subsection{Background} \label{subsection:background}

Boolean networks are a fundamental framework for modelling complex discrete systems, with prominent applications in biology, ecology, and the social sciences~\cite{Gaucherel_2017,Grabisch_2013,Montagud22,Poindron_2021,Rozum_2024}.
A Boolean network represents a network of $n$ interacting entities, where each entity $i \in V = \{1, \dots, n\}$ has a Boolean state $x_i \in \B$, which evolves over time according to a deterministic function $f_i(x_1, \dots, x_n)$ of the current states of the entities.
Mathematically, a Boolean network is simply a mapping $f : \B^n \to \B^n$, which takes an overall configuration of states $x = (x_1, \dots, x_n)$ as input and returns $f(x) = (f_1(x), \dots, f_n(x))$.
The \EM{dynamics} of $f$ is the unlabelled digraph with vertex set $\B^n$, and an arc from each configuration $x$ to $f(x)$. 

The main feature of a Boolean network is its interaction graph. It is a digraph on $V$, with an arc from $i$ to $j$ whenever the local function $f_j$ depends essentially
on the state $x_i$. The interaction graph then represents the architecture of the network, as it captures which components interact which each other. A major problem in the theory of Boolean networks is to understand how the interaction graph, i.e. the structure of the network, influences its dynamics (see \cite{G20} for a review). This question is of interest not only from a theoretical perspective, but also from a practical one, since in many applications, particularly in biology, the interaction graph is known or can be inferred, whereas the underlying dynamics remains largely unknown \cite{N15}. A fixed point of the Boolean network $f$ is a configuration $x$ such that $f(x) = x$. Fixed points are essential, as they represent the stable states of the system, and as such have received significant attention. It is no surprise then that the seminal result in the area, namely Robert's theorem \cite{R80}, is about fixed points: if the interaction graph of $f$ is acyclic, then $f$ converges to a unique fixed point.

Work on the influence of the interaction graph on dynamics broadly follows three avenues. The first one, more classical, aims at inferring properties of the dynamics given a structural property of the interaction graph. Robert's theorem started that avenue; another classical example is the feedback bound: if the interaction graph has a feedback vertex set of size $\tau$, then $f$ has at most $2^\tau$ fixed points \cite{R07}. Recently, \cite{ABGGPRT26} considered the dynamics of networks with an interaction graph with bounded in-degree. A second, more recent avenue of study can be seen as the converse to the first: given a particular dynamics, it aims at inferring the interaction graphs of the Boolean networks with that dynamics. For instance, there exist dynamics on $\B^n$ that can only be produced by dense interaction graphs \cite{bppmr2023}, while other dynamics can be produced by any non-empty interaction graph on $n$ vertices \cite{BPR26}.

Recently, a third avenue of study relating the interaction graph to the dynamics was initiated in \cite{BDPR22} in the case of signed digraphs. In many applications of Boolean networks, some components act as activators/inhibitors of other components. As such, the local functions $f_i$ are monotonically increasing or decreasing in their variables. These monotonic relations are encapsulated by the signed interaction graph: the signed digraph on $V$ with a positive arc (negative arc, respectively) from $i$ to $j$ if $f_j$ is a non-decreasing (non-increasing, respectively) function of $x_i$. The study of the signed interaction graph is an important part of the theory of Boolean networks, which typically followed the first avenue of study mentioned above (for instance, the so-called positive feedback bound in \cite{A08} is the refinement of the feedback bound to positive feedback vertex sets). However, \cite{BDPR22} considers the following problem: given a signed digraph $D$, what is the complexity of deciding whether there exists a Boolean network with signed interaction graph $D$ that satisfies a particular dynamical property? The authors in \cite{BDPR22} focus on the complexity of deciding dynamical properties of the form ``Fewer than $k$ fixed points'' or ``At least $k$ fixed points'', where $k\geq 1$. Remarkably, the first problem is $\NEXPTIME$-complete for general signed digraphs and $\NPoNP$-complete for signed digraphs of bounded in-degree. The second problem is polynomial for $k=1$ and $\NP$-complete for $k\geq 2$, regardless of whether the in-degree is bounded.

\subsection{Contributions} \label{subsection:contributions}

\noindent \textbf{$\phi$-dynamic graphs.}
This paper is the first to explore this third avenue of study in the case of unsigned digraphs. Explicitly, for a given dynamical property $\phi$, we say that a digraph $G$ is \EM{$\phi$-dynamic} if there exists a Boolean network $f$ with interaction graph $G$ that satisfies the dynamical property $\phi$. 
We then consider the complexity of the following decision problem: 

\begin{quote}
\decisionpb
	{$\phi$-Dynamic}
	{a digraph $G$.}
	{is $G$ $\phi$-dynamic?}
\end{quote}

We also consider the same decision problem restricted to digraphs with bounded in-degree:
\begin{quote}
\decisionpb
	{$d$-$\phi$-Dynamic}
	{a digraph $G$ with $\Delta^-(G)\leq d$.}
	{is $G$ $\phi$-dynamic?}
\end{quote}

The relevance to practical problems from biology, where $G$ is known and the dynamics of $f$ is unknown but of interest, is that a negative answer to $\phi$-\textsc{Dynamic} entails that $f$ does not have the dynamical property $\phi$; a positive answer merely fails to exclude this possibility.
For many classical dynamical properties $\phi$, the problem $\phi$-\textsc{Dynamic} is in $\Poly$. Here are a few examples, related to the number of fixed points:
\begin{itemize}
    \item ``At least one fixed point.'' Constant time, as any digraph $G$ is a yes-instance. 
    (Such a network is easy to construct explicitly: let $f_i( x ) = 0$ if $i$ is a source, and $f_i( x ) = \bigwedge_{j \to i} x_j$ otherwise, so the all-zero configuration is a fixed point; there may be others.) %

    This problem is polynomial for signed digraphs \cite{BDPR22} (it reduces to the even dicycle problem, which has a highly non-trivial polynomial solution \cite{McC00,RST99}).

    \item ``At most one fixed point.'' Again, constant time, as any digraph $G$ is a yes-instance.
    (Since the average number of fixed points is exactly $1$ \cite{G20}.)

    This problem is $\NEXPTIME$-complete for signed digraphs \cite{BDPR22}.

    \item ``No fixed points.'' Polynomial, as $G$ is a yes-instance if and only if it is not acyclic.
    (If $G$ is acyclic, then Robert's theorem implies that it is a no-instance. Otherwise, one can explicitly construct a fixed point free Boolean network, see \cite{G20}.)

    This problem is $\NEXPTIME$-complete for signed digraphs \cite{BDPR22}.

    \item ``At least two fixed points.'' Polynomial, as $G$ is a yes-instance if and only if it is not acyclic. 
    (Directly follows from the previous two items.) 
    
    This problem is $\NP$-complete for signed digraphs \cite{BDPR22}.
\end{itemize}
These few examples highlight a strong difference in the complexity landscape between signed and unsigned interaction graphs. Moreover, we see that the same digraph $G$ can be both $\phi$-dynamic and $\neg \phi$-dynamic, e.g. any non-acyclic digraph when $\phi$ is ``At least one fixed point'' and hence $\neg \phi$ is ``No fixed points.''

In this paper, we consider natural properties $\phi$ which can be expressed in FO$( = , \to )$ using the equality sign and the $\to$ relation from the dynamics: $x \to y$ denotes %
that $f(x) = y$.
We note that all the properties above indeed belong to this language; e.g. ``At least one fixed point'' corresponds to $\exists x (x \to x)$.

\medskip 
\noindent \textbf{Complexity of UFP.} We first consider the {\sc Unique Fixed Point (UFP)} problem, i.e. $\phi$-\textsc{Dynamic} with $\phi$ being ``Exactly one fixed point'' or $\exists x \forall y (x \to x) \land ( (y \to y) \implies (x = y) )$ in FO. When the property $\phi$ is given explicitly, we omit the -\textsc{Dynamic} suffix from the decision problem name. 

\begin{quote}
\decisionpb
	{($d$-)Unique Fixed Point (\UFP)}
	{a digraph $G$ (with $\Delta^-(G)\leq d$).}
	{is $G$ UFP-dynamic?}
\end{quote}

We prove the hardness of $\UFP$, thus establishing, for the first time, a property $\phi$ for which $\phi$-\textsc{Dynamic} is computationally hard. More precisely, we obtain:

\begin{theoremrep}\label{thm:UFP}
{\UFP} is $\NEXPTIME$-complete and {\dUFP} is $\NPoNP$-complete for all $d\geq 2$. 
\end{theoremrep}

The hardness of $\UFP$ is surprising in several ways. 
First, as seen above, many related $\phi$-\textsc{Dynamic} problems are in $\Poly$. 
Second, as we shall prove in this paper, $\UFP$ can be decided in linear time when the input digraph is strongly connected.
Third, the complexity of the Unique Fixed Point problem is also specific to the Boolean case. In the non-Boolean case, for any digraph $G$ and any $q \ne 2$, the property is universal: if $q \ge 3$ then there is always a function $f : \{0, \dots, q-1\}^n \to \{0, \dots, q-1\}^n$ with interaction graph $G$ and a unique fixed point \cite{GR16}, and if $q = 1$ then every function trivially has a unique fixed point.

\medskip

\noindent \textbf{Further dynamical properties.} We then consider a wide family of dynamical properties which we call Period Piece$(\Lambda, \Sigma)$ and define as follows. 
Let $\Lambda$ be a non-empty finite set of positive integers and $\Sigma$ be a non-empty finite set of finite permutations $\sigma$ such that all the cycles of $\sigma$ have lengths in $\Lambda$. 
Let $\Per_\Lambda( f )$ be the set of periodic points of $f$ on a cycle whose length belongs to $\Lambda$. 
We say that $f$ is \EM{Period Piece$(\Lambda, \Sigma)$} if $f$ restricted to $\Per_\Lambda( f )$ is conjugate to some $\sigma \in \Sigma$.

\begin{quote}
\decisionpb
	{($d$-)$\PeriodPiece(\Lambda, \Sigma)$}
	{a digraph $G$ (with $\Delta^-(G)\leq d$).}
	{is $G$ Period Piece$(\Lambda, \Sigma)$-dynamic?}
\end{quote}

The Period Piece$(\Lambda, \Sigma)$ property includes many dynamical properties about periodic points. 
For instance, if $\Lambda = \{ 1 \}$ and $\Sigma = \{ \id_1, \dots, \id_k \}$, it is equivalent to ``Between $1$ and $k$ fixed points.'' 
Moreover, if $\Lambda = \{ \ell \}$ and $\Sigma = \{ C_\ell \}$ where $C_\ell$ is a cyclic permutation on $\ell$ elements, it is equivalent to ``Unique cycle of length $\ell$'';
on the other hand, if $\Lambda=\{1,\ldots, \ell\}$ and $\Sigma=\{C_\ell\}$ then it is equivalent to ``Unique cycle of length $\ell$ and no cycle of length less than $\ell$''.

In addition to UFP's appeal as a natural problem for Boolean networks, it acts as the fulcrum of this work: indeed, our hardness proofs for $d$-\PeriodPiece{} and \PeriodPiece{} are both reductions from ($d$-)UFP.
We prove the hardness of $\PeriodPiece(\Lambda, \Sigma)$ for all possible $\Lambda$ and $\Sigma$. Let $\N^* = \{ 1, 2, 3, \dots \}$ and denote by $S_\Lambda$ the set of finite permutations $\sigma$ such that all the cycles of $\sigma$ have lengths in $\Lambda$. 

\begin{theoremrep} \label{theorem:perlambda}
For every finite non-empty sets $\Lambda \subseteq \N^*$ and $\Sigma \subseteq S_\Lambda$, $\PeriodPiece(\Lambda, \Sigma)$ is $\NEXPTIME$-complete and $d$-$\PeriodPiece(\Lambda, \Sigma)$  is $\NPoNP$-complete for all $d \ge 2$.
\end{theoremrep}

 \medskip

\noindent \textbf{Comparison with signed digraphs.} We would like to emphasize the distinction between our work and that in \cite{BDPR22}. Signed digraphs are essentially easier to handle in order to derive complexity lower bounds because of the nature of the local functions and of their sheer numbers. As an illustrative example, consider a vertex $u$ of the signed interaction graph with a positive in-neighbor $v$ and a negative in-neighbor $w$. Then there are only two possible choices for the local function $f_u$: either $f_u( x_v, x_w ) = x_v \land \neg x_w$ or $f_u( x_v, x_w ) = x_v \lor \neg x_w$. Both possible local functions are canalyzing: for $f_u( x ) = x_v \land \neg x_w$, if $x_v = 0$, then $f_u( x ) = 0$; if $x_w = 1$, then $f_u( x ) = 0$ (similar canalyzations occur for $f_u( x_v, x_w ) = x_v \lor \neg x_w$). These canalyzations make it possible to track the values of components in any fixed point of $f$. In contrast, when dealing with unsigned digraphs, if $u$ has two in-neighbors $v$ and $w$, then there are 10 choices for $f_u$, some of which are not canalyzing, like $f_u( x_v, x_w ) = x_v + x_w$. As a result, it is much more difficult to handle the wealth and variety of Boolean networks that have a given interaction graph, compared to the more restricted and well-behaved class of networks that have a given signed interaction graph. 
This in turn requires us to develop significantly different tools for our main results compared to those in \cite{BDPR22}: in contrast to the proof techniques in that work, where an explicit hardness reduction is given directly, in our setting we build a theory of so-called ``odd-systems'' to support our constructions. 

\medskip

\noindent
\textbf{Contribution to the $\NEXPTIME$ landscape.} 
In addition to being natural for the Boolean network community and justified from the perspective of biological applications, the {\PeriodPiece} problem is, up to the authors' knowledge, the first naturally-defined language on digraphs which is complete for $\NEXPTIME$.
Typically, $\NEXPTIME$-complete problems either take as input a circuit which is a succinct representation of some object \cite{P94}, or include the binary representation of some integer \cite{GI09-tiling-NEXP}.  
We note that it is trivial to construct an \textit{unnaturally}-defined problem on digraphs which is $\NEXPTIME$-complete: for example, a variant of \textsc{Succinct SAT} where the input circuit is given (via some encoding) as a digraph.
The $\phi$-{\sc Dynamic} problem, on the other hand, is a ``true'' digraph problem, with its input given explicitly as a tuple $(V,E)$.
{\PeriodPiece} thus constitutes a vast family of $\NEXPTIME$-complete natural languages on digraphs.

\newpagexorapx

\section{Preliminaries}
\subsection{Definitions and notations}

A \EM{configuration} is a function $x$ from a finite set (of components) $V$ to $\B$. Such a configuration $x$ is a configuration \EM{on} $V$ and given $v\in V$ we write $x_v$ instead of $x(v)$. The set of configurations on $V$ is denoted $\B^V$. We perform all sums over $\B$ modulo $2$. We say that $x$ is \EM{odd} if there is an odd number of $v\in V$ such that $x_v=1$, or in other words if $\sum_{v \in V} x_v = 1$; we say $x$ is \EM{even} otherwise. Given $I\subseteq V$, the restriction of $x$ on $I$ is denoted \EM{$x_I$}. By extension, for any set of configurations $X \subseteq \B^V$, we denote $X_I = \{x_I: x \in X\}$. A configuration $y$ \EM{extends} $x$, and $x$ is a \EM{subconfiguration} of $y$, if $x = y_I$ for some $I\subseteq V$. The configuration on $V$ which always returns $0$ is denoted $\ZERO_V$ or $\ZERO$ if $V$ is clear from the context. Given $I\subseteq V$, \EM{$e_I$} denotes the configuration on $V$ such that $(e_I)_v=1$ if $v\in I$ and $(e_I)_v=0$ otherwise; here $V$ will always be clear from the context. 

The vertex set of a digraph $G$ is denoted \EM{$V(G)$}. A digraph \EM{on} $I$ is a digraph $G$ with $V(G)=I$. We abusively write \EM{$\B^G$} instead of $\B^{V(G)}$ and say that a configuration in $\B^G$ is a configuration on $G$. Similarly, we write \EM{$\ZERO_G$} instead of $\ZERO_{V(G)}$. 

A \EM{Boolean network} is a mapping $f : \B^V \to \B^V$.
A Boolean network can be viewed as a tuple $f = ( f_v : v \in V )$ where $f_v( x ) = f( x )_v$ for any $v \in V$.
We use similar notations as for configurations, and denote, e.g. $f_I = ( f_v : v \in I )$ for any $I \subseteq V$.
The \EM{dynamics} of $f$ is the digraph $D_f$ with vertex set $\B^V$ and arcs $\{ x \to y : f(x) = y \}$.
A \EM{fixed point} of $f$ is a configuration $x \in \B^V$ such that $f(x) = x$.

The \EM{interaction graph} of the Boolean network $f$ is the digraph $G_f$ on $V$ such that there is an arc from $i$ to $j$ if and only if $f_j$ depends essentially on $x_i$,  
i.e. there exists $x$ such that $f_j(x)\ne f_j(x + e_i)$. 
We say $f$ is a network \EM{over} $G_f$.
Therefore, each local function $f_j$ only depends on $x_N$, where $N$ is the in-neighborhood of $j$ in $G_f$.
We then use the shortcut notation $f_j( x ) = f_j( x_A )$ whenever $N \subseteq A$.
For a digraph $G$, the set of Boolean networks $f$ whose interaction graph is $G_f = G$ is denoted by $F( G )$.

A transformation is a function $f:X\to X$ where $X$ is a finite set, called the \EM{domain} of $f$. Two transformations $f,h$ with domains $X,Y$ are \EM{conjugate} if there exists a bijection $\pi:X\to Y$ such that $f=\pi^{-1}\circ h\circ \pi$. Equivalently, $f$ and $h$ are conjugate if and only if the digraphs associated with them, $D_f$ and $D_h$, are isomorphic. 
We say a property $\phi$ of Boolean networks is a \EM{dynamical property} if whenever $D_f \models \phi$, $D_g \models \phi$ for all $g$ conjugate to $f$. In other words, $\phi$ is a property of the isomorphism class of $D_f$. 

In this paper, we focus on dynamical properties $\phi$ which can be expressed in First Order logic. We consider the language $\mathrm{FO}( = , \to )$ with the equality and the assignment symbols, where $x \to y$ means that $f(x) = y$. We further consider the fragments of FO with bounded alternations: let $\mathrm{FO}\Sigma_i$ be the set of FO formulas expressible as the alternation $\phi \equiv \exists X^1 \forall X^2 \exists X^3 \dots Q X^i \psi( X^1, \dots, X^i )$ and $\mathrm{FO}\Pi_i$ be the set of FO formulas expressible as the alternation $\phi \equiv \forall X^1 \exists X^2 \forall X^3 \dots Q' X^i \psi( X^1, \dots, X^i )$. It is easily seen that $\mathrm{FO}\Sigma_i$ and $\mathrm{FO}\Pi_i$ are closed under conjunction and disjunction.

Also, we denote $\Ns$ the set $\N = \{1,2,3\ldots\}$ and we recall that for a permutation $\sigma$ of $k$ elements, $k$ is the \EM{degree} of $\sigma$.

\subsection{Basic results}

For any dynamical property $\phi$, and any digraph $G$, we say $G$ is \EM{$\phi$-dynamic} if there exists $f \in F(G)$ such that $D_f \models \phi$. We are particularly interested in the dynamical property ``Unique Fixed Point (UFP),'' which is expressed in $\mathrm{FO}\Sigma_2$ as $\exists x \forall y (x \to x) \land [ (y \to y) \implies (y = x) ]$. We begin by classifying the strongly connected UFP-dynamic digraphs.

\begin{theoremrep}\label{thm:strong-not-cycle-implies-ufp}
A strongly connected digraph is UFP-dynamic if and only if it is not a cycle. 
\end{theoremrep}

\begin{proofsketch}
If $G$ is a cycle, then any $f \in F(G)$ has either zero or two fixed points \cite{DNS12}.
On the other hand, every strongly connected digraph can be obtained from a cycle by successive additions of ears \cite[Theorem 7.2.2]{BG08}. 
As such, we are able to recursively construct a network $f \in F(G)$ with a unique fixed point thanks to the ear decomposition of the digraph.  
\end{proofsketch}

\begin{proof}
If $G$ is a cycle, then any $f \in F(G)$ has either zero or two fixed points \cite{DNS12}. We now prove the converse.

Adding an ear in a directed digraph from a vertex $a$ to a vertex $b$, not necessarily distinct, consists either in adding an arc from $a$ to $b$ if it does not already exist, or in adding a path whose vertices are new, say from $a'$ to $b'$ (possibly $a'=b'$), together with an arc from $a$ to $a'$ and from $b'$ to $b$. In the resulting digraph, $b$ has a new in-neighbor, which is $a$ in the first case, and $b'$ in the second; in the second case $b'$ is linear and there is a path from $a$ to $b'$ whose internal vertices are linear. 

Every strongly connected digraph can be obtained from a cycle by successive additions of ears \cite[Theorem 7.2.2]{BG08}. Consequently, a strongly connected digraph $G$ which is not a cycle can be obtained from a strongly connected digraph $H$ by adding one ear, say from $a$ to $b$. Let $c$ be the corresponding new in-neighbor of $b$. Let $N_v$ be the in-neighbors of a vertex $v$ in $G$. Let $f\in F(G)$ be defined by
\begin{itemize}
\item
$f_v(x)=\bigwedge_{u\in N_v} x_u$ for all vertices $v\neq b$, and 
\item
$f_b(x)=\neg x_{c}\land \bigwedge_{u\in N_b\setminus \{c\}} x_u$.
\end{itemize} 
It is clear that $f(\ZERO)=\ZERO$. Suppose that $f(x)=x$. First assume that $x_a=0$. Since $H$ has a path from $a$ to $b$ (which does not contain $c$) we have $x_b=0$. Suppose now that $x_a=1$. If $a=c$ then $x_b=0$. Otherwise, there is a path from $a$ to $c$ whose internal vertices are linear, and thus $x_{c}=1$ (since $c$ is also linear), so that $x_b=0$. Thus, in every case, $x_b=0$. Since $G$ has a path from $b$ to all the other vertices, we deduce that $x=\ZERO$. Thus $\ZERO$ is the unique fixed point of $f$.  
\end{proof}

Given a digraph $G$, we denote by \EM{$F_\ZERO(G)$} the set of $f\in F(G)$ such that $f(\ZERO)=\ZERO$, and we denote by \EM{$UFP_\ZERO( G )$} the set of $f \in F( G )$ such that $\ZERO$ is the unique fixed point of $f$; by definition, $UFP_\ZERO( G ) \subseteq F_\ZERO( G ) \subseteq F( G )$.

\begin{lemmarep}\label{lem:UFP0}
Let $G$ be a digraph. Then $F_\ZERO(G)$ is not empty. Furthermore, if $G$ is UFP-dynamic, then $UFP_\ZERO(G)$ is non-empty. 
\end{lemmarep}

\begin{proof}
Let any $h\in F(G)$. Let $f:\B^G\to\B^G$ defined by $f(x)=h(x)+h(\ZERO)$. We easily check that $f\in F(G)$. Furthermore, $f(\ZERO)=h(\ZERO)+h(\ZERO)=\ZERO$. Let $h\in F(G)$ with a unique fixed point, say $y$. Let $f:\B^G\to\B^G$ defined by $f(x)=h(x+y)+y$. We easily check that $f\in F(G)$. Furthermore, if $f(x)=x$ then $h(x+y)+y=x$, that is, $h(x+y)=x+y$. Since $y$ is the unique fixed point of $h$, we deduce that $x=\ZERO$. Thus $\ZERO$ is the unique fixed point of $f$. 
\end{proof}

\newpagexorapx

\section{Complexity upper bounds}

We begin by establishing complexity upper bounds for our problems. In the case of $\phi$-\textsc{Dynamic} we are able to apply Courcelle's theorem to obtain containment in $\NEXPTIME$ for any property $\phi$ expressible in monadic second order logic (MSO2), which allows quantification over sets of vertices and over sets of edges and thus is strictly more expressive than, e.g., $\mathrm{FO}(=,\to)$. For the case of $d$-$\phi$-\textsc{Dynamic} we restrict ourselves to properties expressible in First Order logic. These upper bounds apply to all properties $\phi$ which we consider. 

\begin{theoremrep}\label{thm:in-nexptime}
    Let $\phi$ be a property expressible in MSO2. Then $\phi$-\textsc{Dynamic} is in $\NEXPTIME$. 
\end{theoremrep}

\begin{proofsketch}    
The proof constructs a Turing machine $M$ with runtime $O(2^{\mathrm{poly}(n)})$ such that the language $L$ of $\phi$-dynamic digraphs is given by $L = \{w | \exists c : M \textrm{ accepts } \langle w,c \rangle\}$,
    where $w$ encodes the graph $G$ and $c$ encodes the network $f \in F(G)$.
    In order to decide whether $D_f \models \phi$, we realize that $D_f$ has bounded treewidth and hence apply Courcelle's theorem.
\end{proofsketch}

\begin{proof}
    Recall that the decision problems we are interested in can be expressed as languages of the following form:
    $$
    L_\phi = \{G | \exists f \in F(G) : f\models \phi\}
    $$
    where $\phi$ is a formula in MSO2 and $G$ is a digraph.
    
    Recall that a language $L$ is in $\NEXPTIME$ if and only if there is a Turing machine $M$ with runtime $O(2^{\mathrm{poly}(n)})$ such that:
    $$
    L = \{w | \exists c : M \textrm{ accepts } \langle w,c \rangle\}
    $$
    where $c$ is called the certificate (we refer the interested reader to, e.g., \cite{P94}) and $n$ is the size of the input $w$ (in our case the number of vertices in $G$).

    For our purposes, $w$ is always a digraph $G$. The intention is that $c$ shall represent a function $f: \{0,1\}^{V(G)} \to \{0,1\}^{V(G)}$. 

    We shall evidence that for any MSO2 sentence $\phi$, $L_\phi$ is in $\NEXPTIME$ by describing the machine $M_\phi$ which yields $L_\phi$ when plugged into the definition above. 
    
    As a first step, $M_\phi$ verifies that the function $f$ (as encoded in $c$) is indeed in $F(G)$ by computing $G_f$. This is straightforwardly doable in exponential time -- checking whether $f_i$ depends essentially on $x_j$ (and thus whether $G_f$ has an arc from $j$ to $i$) can be done by brute force in time $O(2^n)$ for each $i,j \in [n]$, so verifying that $(i,j) \in E(G) \iff (i,j) \in E(G_f)$ is doable in time $O(n^2\cdot 2^n)$. (If $c$ does not encode a function $f\in F(G)$ then $M_\phi$ rejects.)
    
    Next, $M_\phi$ decides whether $D_f \models \phi$. The functional digraph $D_f$ of $f$ can be  represented explicitly in space $O(n\cdot 2^n)$ (such a representation would be computed in the previous step). Also, $D_f$ has bounded treewidth: each connected component of $D_f$ contains at most one cycle, and thus the component (and $D_f$ as a whole) has treewidth at most $2$. 
    
    Consequently, by applying Courcelle's theorem (see \cite{DF13}, chapter 13) $M_\phi$ can decide whether $D_f \models \phi$ in time $|D_f| \cdot g(\mathrm{tw}(D_f)+|\phi|)) \in O(|D_f|)$ for some function $g : \N \to \N$. Since $|D_f|=2^n$, for any fixed $\phi$ the runtime is simply $O(2^{\mathrm{poly}(n)})$. The result follows.
\end{proof}

\begin{theoremrep}
    Let $\phi$ be a property in $\mathrm{FO}\Sigma_i$. Then $d$-$\phi$-{\sc Dynamic} is in $\Sigma_i^P$. 
\end{theoremrep}

\begin{proofsketch}
We first note that when the in-degree is bounded by $d$, any Boolean network $f \in F(G)$ can be represented as a collection of $n$ lookup tables each of size $O(2^d)$.
Thanks to that representation, one can evaluate $f$ at $x$ in linear time $O( n \cdot d )$ on input $x$.

Let $\phi$ be represented by $\exists X^1 \forall X^2 \exists X^3 \ldots Q X^i \psi(X^1,X^2,X^3,\ldots,X^i)$ where $\psi$ is quantifier-free. The proof constructs a polynomial Turing machine $M$ such that the following characterizes the language $L$ of $\phi$-dynamic digraphs with in-degree at most $d$:
    $$
    L = \{w | \exists c^1 \forall c^2 \exists c^3 \ldots Q c^i : M \textrm{ accepts } \langle w, c^1, c^2, \ldots, c^i\rangle\},
    $$
    where $w$ encodes the graph $G$, $c^1$ encodes a tuple $\langle f, a^1 \rangle$ and for $1<j\le i$, $c^j$ encodes $a^j$, and each $a^j$ an assignment to $X^j$.
\end{proofsketch}

\begin{proof}
    We use $Q^i$ to denote $\exists$ if $i$ is odd and $\forall$ otherwise.     
    Since $\phi$ is in $\mathrm{FO}\Sigma_i$ we shall represent it as
    $$
    \exists X^1 \forall X^2 \exists X^3 \ldots Q^i X^i \psi(X^1,X^2,X^3,\ldots,X^i) 
    $$
    where $\psi$ is quantifier-free and each $X^j=(x^{j,1},\ldots, x^{j,\ell(j)})$ is a finite sequence of $\ell(j)$ variables.
    
    A language $L$ is in $\Sigma_i^P$ if and only if there is a polynomial Turing machine $M$ such that the following characterizes $L$ (\cite{P94}, Chapter 17):
    $$
    L = \{w | \exists c^1 \forall c^2 \exists c^3 \ldots Q^i c^i : M \textrm{ accepts } \langle w, c^1, c^2, \ldots, c^i\rangle\}.
    $$
    
    Since we are concerned with the complexity of $d$-$\phi$-\textsc{Dynamic} for some $d\in \Ns$, $\phi \in \mathrm{FO}\Sigma_i$, which is a subset of the digraphs of maximum in-degree $d$, for us $w$ is always a digraph $G$ of maximum in-degree $d$. 

    The key claim for this proof is the following:

    \begin{claim}\label{clm:in-deg-2-linear-eval}
        Let $G$ a digraph of maximum in-degree $d$ on $n$ vertices and $f \in F(G)$. 
        Then there is an algorithm $\mathcal{A}_f$ of size $O(n \cdot 2^d)$ which returns $f(x)$ in time $O(n\cdot d)$ on input $x$.
    \end{claim}
    \begin{claimproof}
        Let $f(x)=(f_1(x),\ldots, f_n(x))$. Then for each $j$, $f_j$ is fully described by a lookup table $\mathrm{LUT}_j$ with $2^d$ entries, indexed by $\{0,1\}^{N^-(j)}$. Consider the algorithm $\mathcal{A}_f$ which has access to each $\mathrm{LUT}_j$ (using space $2^d \cdot n$) and, on input $x$:
        \begin{enumerate}
            \item Produces the sets $N^-(j)$ for each $j$ in time $|E(G)|\le n \cdot d$. 
            \item Computes $x_{N^-(j)}$ for each $j$ in time $O(d)$ each and $O(n \cdot d)$ in total.
            \item Returns $(\mathrm{LUT}_1[x_{N^-(1)}], \ldots, \mathrm{LUT}_n[x_{N^-(n)}])$.
        \end{enumerate}
        Then $\mathcal{A}_f$ satisfies the conditions of the lemma.
    \end{claimproof}

    The intention is that $c^1$ encodes a tuple $\langle f, a^1 \rangle$ and for $1<j\le i$, $c^j$ encodes $a^j$, with $f$ given as (an encoding of) the algorithm $\mathcal{A}_f$ as in Claim \ref{clm:in-deg-2-linear-eval} and each $a^j=(a^{j,1},\ldots,a^{j,\ell(j)})$ an assignment to $X^j$.  

    As for the proof of Theorem \ref{thm:in-nexptime}, we shall describe a machine $M_{d,\phi}$ which defines the language $d$-$\phi$-\textsc{Dynamic} when plugged into the definition above. On input $\langle w, c^1, \ldots, c^i\rangle$, $M_{d,\phi}$ first verifies that each element of the input encodes the appropriate object(s); for example, if $w$ does not encode a digraph $G$ of maximum degree at most $d$ then $M_{d,\phi}$ rejects early. 
    Likewise if $c^j$ does not encode an assignment $a^j$ to the (existentially-quantified) variables in $X^j$ for some odd $j$.
    For even $j$, $c^j$ is universally quantified and we do the opposite: if $c^j$ does not encode an assignment $a^j$ to the variables in $X^j$ then $M_{d,\phi}$ accepts. 
    Thus, we may henceforth assume that the input is a tuple of the form described above, and make use of the notation we had introduced to explain how $M_{d,\phi}$ processes the input.

    $M_{\phi,d}$ makes a linear pass over the quantifier-free formula $\psi$. Each variable in $\psi$ is $x^{j,k}$ for some $j,k \in \Ns$, and has corresponding assignment $a^{j,k} \in \{0,1\}^{V(G)}$. To simplify the presentation of the proof, we consider variables $x,y$ in $\psi$ which are assigned $u,v$ respectively. Thus two types of primitives can be encountered during the linear pass:
    \begin{itemize}
        \item For a primitive $(x = y)$, the value \textrm{true} is substituted by $M_{\phi,d}$ if $x=y$ and \textrm{false} is substituted otherwise.
        \item For a primitive $(x \to y)$, $f(x)$ is evaluated using $\mathcal{A}_f$ in time $O(n \cdot 2^d)$ as a subroutine in $M_{\phi,d}$ and the value \textrm{true} is substituted if $f(x)=y$ and \textrm{false} is substituted otherwise.
    \end{itemize}
    Evaluating whether $\psi$ is satisfied by the given assignment is thus feasible in time $O(|\psi|\cdot n \cdot 2^d)$. 
    
    Correctness of the algorithm is straightforward. If there exists some $f\in F(G)$ such that $D_f \models \phi$ then 
    $M_{d,\phi}$ accepts $\langle G, \langle \mathcal{A}_f, a^1 \rangle, a^2, \ldots, a^i\rangle$ exactly according to the characterization above. 
    On the other hand, it is clear that if $M_{d,\phi}$ accepts a graph $G$ then there exists some $f \in F(G)$ such that $D_f \models \phi$.

    The result follows.
\end{proof}

We state the following corollary to make explicit the result when $i\in \{1,2\}$:

\begin{corollary}\label{cor:in-np-or-nponp}
    $d$-$\phi$-\textsc{Dynamic} is in $\NP$ (resp. $\NPoNP=\Sigma_2^P$) for any $\phi$ in $\mathrm{FO}\Sigma_1$ (resp. $\mathrm{FO}\Sigma_2$). 
\end{corollary}

\newpagexorapx
\section{Cycle functionality} \label{section:cycle_functionality}

In this section, we investigate the dynamical behavior of a single chordless cycle in a larger interaction graph. 
We begin by proving a slight generalization of Robert's theorem, and then develop the concept of what we shall call the context of a cycle.

Let $G$ be a digraph. Robert \cite{R80,R86,R95} proved that if $G$ is acyclic then, for every $f\in F(G)$, $f^n$ is a constant function, where $n = |V|$ and $f^n = f \circ \cdots \circ f$, so that every configuration on $G$ converges to the same fixed point in at most $n$ iterations. We will often use a local version of this theorem, and for that we need some definitions. 

Let $I,J\subseteq V(G)$ and $f\in F(G)$. We say that $z\in\B^I$ is \EM{fixed by} $f$ if $f(x)_I=z$ for all $x\in\B^G$ extending $z$. For instance, if all the in-neighbors of vertices in $I$ are in $I$, and if $x$ is a fixed point of $f$, then $x_I$ is fixed by $f$. 
Let $z\in\B^I$ and $z'\in \B^J$. We say that $z$ \EM{converges} to $z'$ if, for every $x\in\B^G$ extending $z$, there exists $k \ge 1$ such that $f^k(x)$ extends $z'$. Let $\EMM{\epsilon}\in\B^\emptyset$ be the empty configuration; by convention, every configuration extends $\epsilon$, and thus $\epsilon$ is fixed by $f$. Robert's theorem is the extremal case of the following lemma where $I=\emptyset$ (which implies $z=\epsilon$) and $J=V(G)$. The proof of the following lemma is by induction on $|J|$, and is a natural generalization of Robert's theorem.

\begin{lemmarep}\label{lem:robert}
Let $G$ be a digraph and $I,J\subseteq V(G)$ disjoint. Let $f\in F(G)$ and suppose that $z\in\B^I$ is fixed by $f$. If $G[J]$ is acyclic and all the in-neighbors of vertices in $J$ are in $I\cup J$, then there exists a unique $z'\in\B^{I\cup J}$ extending $z$ which is fixed by $f$, and $z$ converges to~$z'$.
\end{lemmarep}

\begin{proof}
We proceed by induction on $|J|$. The claim immediately holds if $|J|=0$ (there is nothing to prove). Suppose that $|J|\geq 1$. Since $G[J]$ is acyclic, it has a sink, say $a$. Let $J'=J\setminus \{a\}$ and $K=I\cup J'$. Since $G[J']$ is acyclic and all the in-neighbors of vertices in $J'$ are in $K$, by induction, there exists a unique $z'\in\B^K$ extending $z$ which is fixed by $f$, and $z$ converges to~$z'$. Let $z''\in\B^{I\cup J}$ extending $z'$ defined by $z''_a = f_a( z' )$.

If $x\in\B^G$ extends $z'$ then $f(x)_{K}=z'$ and, since all the in-neighbors of $a$ are in $K$, $f_a(x)=f_a(\sum_{b\in K}x_be_b)=f_a(\sum_{b\in K}z'_be_b)=z''_a$. We deduce that $z''$ is the unique configuration on $I\cup J$ extending $z'$ which is fixed by $f$ and that $z'$ converges to $z''$. Hence $z$ converges to $z''$. To prove the uniqueness, suppose that $y\in\B^{I\cup J}$ extends $z$ and is fixed by~$f$. Since all the in-neighbors of vertices in $J'$ are in $K$, $y_K$ is fixed by~$f$. Since $y_K$ extends $z$, we obtain $y_K=z'$. So $y$ extends $z'$ and thus $y=z''$. 
\end{proof}

As mentioned above, if the interaction graph of a Boolean network is acyclic, then the network converges towards a fixed point. If, on the other had, the interaction graph is a cycle, then there are two cases; one if the cycle is so-called positive, and one if it is negative. 
Let us illustrate these with some examples. Let $H$ be a cycle graph with arcs $i$ to $i+1 \mod n$ for $0 \le i \le n-1$, and let $h \in F(H)$. 
Any local function $h_i$ is of the form $h_i( x ) = x_{i-1} + \alpha_i$; if $\alpha_i = 0$, then $h_i$ copies $x_{i-1}$, and if $\alpha_i = 1$, then $h_i$ copies the negation of $x_{i-1}$. 
Let $\alpha = (\alpha_0, \dots, \alpha_{n-1}) \in \B^n$. If $\alpha = \ZERO$, then $h$ is a cyclic shift; it has two fixed points, namely the constant configurations $\ZERO$ and $\ONE$. The dynamics of $h$ remain the same (up to conjugation) whenever $\alpha$ is even: in this case $h$ is a positive cycle. Conversely, if $\alpha$ is a unit configuration (i.e. there is a unique $i$ such that $\alpha_i = 1$), then $h$ has no fixed points. Once again, the same can be said whenever $\alpha$ is odd: in that case, $h$ is a negative cycle. Much more is known about cycles; we refer the reader to \cite{DNS12} for a comprehensive analysis of the dynamics of positive and negative cycles.

If a chordless cycle in a graph has in-neighbors outside of its vertex set, then given the configuration of its neighbors, the Boolean network restricted to the cycle, say $h$, can act in three ways: $h$ acts like a positive cycle, $h$ acts like a negative cycle, or $h$ acts like an acyclic network. It is easy to see how the first two cases can occur. For an example of the third one, suppose $v$ is the other in-neighbor of the vertex $i$ of the cycle, and let $h_i( x_{i-1}, x_v ) = x_{i-1} \land x_v$. Then when $x_v = 0$, we have $h_i( x_{i-1}, 0 ) = 0$ is a constant function. As such, the interaction graph of $h$ restricted to $x_v = 0$ becomes acyclic, as the arc from $i-1$ to $i$ is removed. We note that it only takes one arc to be removed for $h$ to become acyclic.

Let us now formalize the intuition. Let $G$ be a digraph and $f\in F(G)$. Suppose that $G$ contains a chordless cycle $C$, with vertex set $J$, and let $N = N^-(J) \setminus J$ be the in-neighbors of the vertices in $J$ that do not belong to $J$. Let us label the vertices of $C$ from $a_0$ to $a_n$ following the order; for convenience $a_{-1}=a_n$. A \EM{context} for $C$ is a configuration on a set $I$ of vertices in $G$ with $N\subseteq I$ and $I\cap J=\emptyset$. Given a context $z$ for $C$, one of the following three cases holds: 
\begin{itemize}
\item
There exists $0\le i\le n$ and $\beta_i$ such that $f_{a_i}(x) = \beta_i$ for all $x\in\B^G$ extending~$z$. In that case, we say that $z$ is an \EM{acyclic context} and that $a_i$ is \EM{fixed} to $\beta_i$ by the context.

\item
For all $0\leq i\leq n$, there exists $\alpha_i$ such that $f_{a_i}(x)=x_{a_{i-1}}+\alpha_i$ for all $x\in\B^G$ extending $z$. In that case, we say that $z$ is a \EM{functionality context}, and that $\alpha=(\alpha_0,\dots,\alpha_n)$ is the \EM{signature} of $C$ in that context. 
\begin{itemize}
    \item If $\alpha$ is even, then $z$ is a \EM{positive context}.

    \item If $\alpha$ is odd, then $z$ is a \EM{negative context}.
\end{itemize}
\end{itemize}
See Figure \ref{fig:context} for an illustration.

\begin{figure}[!ht]
    \centering
    \includegraphics[page=5, width=\linewidth]{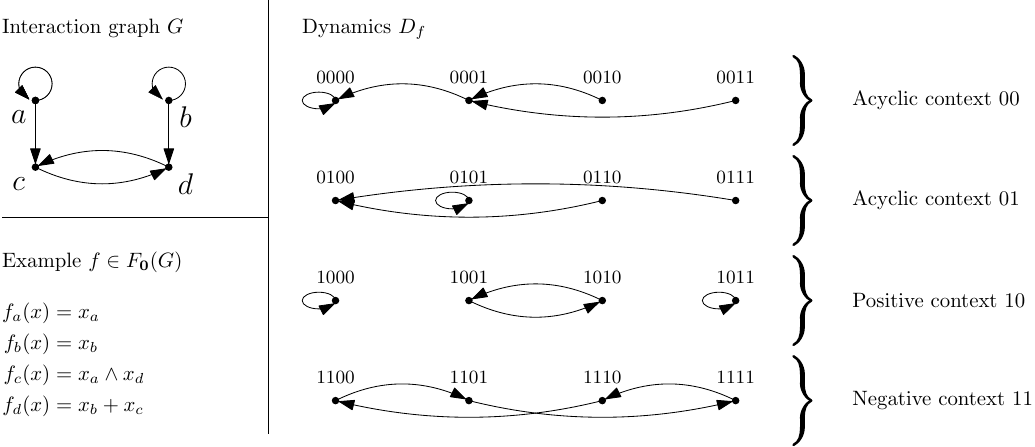}
    \caption{\textit{Top left:} a digraph $G$ with cycle $C = c \to d \to c$. \textit{Bottom left:} a Boolean network $f$ over $G$. \textit{Right:} the dynamics $D_f$, separated by context for $C$.}
    \label{fig:context}
\end{figure}

The last two points of the following lemma can be obtained from exhaustive results about the dynamics of isolated cycles \cite{DNS12}, but since the proof is short, we do not rely on these results. We shall exclusively apply this lemma in settings where no two cycles in the graph $G$ intersect in any vertex.

\begin{lemmarep}\label{lem:acyclic-positive}
Let $G$ be a digraph and $f\in F(G)$. Suppose that $G$ contains a chordless cycle $C$ of length $\ell$. Let $z\in \B^I$ be a context for $C$ fixed by $f$, and let $J$ be the vertex set of $C$. 
\begin{itemize}
\item
If $z$ is an acyclic context, there exists a unique configuration $z'$ on $I\cup J$ extending $z$ which is fixed by~$f$, and $z$ converges to $z'$. 
\item
If $z$ is a positive context, two distinct configurations on $I\cup J$ extending $z$ are fixed by~$f$.
\item
If $z$ is a negative context, then $z$ converges to a periodic configuration whose period is $q\cdot d$ for some $q,d \in \mathbb N^*$ such that $d\mid 2\ell$ and $d\nmid \ell$.
\end{itemize}  
\end{lemmarep}

\begin{proof}
Let $n=\ell-1$, and let us label the vertices of $C$ from $a_0$ to $a_n$ following the order. 

Suppose that $z$ is an acyclic context, and suppose, without loss, that $a_0$ is fixed to $\beta_0$ by the context. Let $z'$ be the configuration on $I' = I \cup \{ a_0 \}$ extending $z$ such that $z'_{a_0}=\beta_0$. Then $z'$ is fixed by $f$. Since $G[J\setminus \{ a_0 \}]$ is acyclic (actually a path), and since all the in-neighbors of vertices in $J\setminus\{ a_0 \}$ are in $I\cup J$, we deduce from Lemma~\ref{lem:robert} that there exists a unique configuration $z'$ on $I\cup J$ extending $z$ which is fixed by~$f$, and $z$ converges to $z'$. 

Suppose that $z$ is a positive context, and let $\alpha$ be the corresponding signature. Let $c\in\{0,1\}$, and let $z^c$ be the configuration on $I\cup J$ extending $z$ obtained by setting $z^c_{a_0}=c$ and successively $z^c_{a_i}=z^c_{a_{i-1}}+\alpha_i$ for $1\leq i\leq n$. Let $x$ be a configuration on $G$ extending $z^c$. Then $f(x)_I=z_I$ since $x$ extended $z$. Next, since $x_{a_n}=z^c_{a_n}=c+\alpha_1+\cdots+\alpha_n$, we have $f_{a_0}(x)=x_{a_n}+\alpha_0=z^c_{a_n}+\alpha_0=c+\alpha_0+\alpha_1+\cdots+\alpha_n$, and thus $f_{a_0}(x)=c=z^c_{a_0}$ since $\alpha$ is even. Finally, for $1\leq i\leq n$, we have $f_{a_i}(x)=x_{a_{i-1}}+\alpha_i=z^c_{a_{i-1}}+\alpha_i=z^c_{a_i}$. Thus $z^0$ and $z^1$ are fixed by $f$.
Suppose that $z$ is a negative context, and let $\alpha$ be the corresponding signature. Then $z$ converges to a periodic configuration $x$ of $f$, and $f^k(x)$ extends $z$ since $z$ is fixed by $f^k$ for all $k \ge 1$. Since $z$ is a functional context, we have $f_{a_0}(x)=x_{a_n}+\alpha_0$, and we obtain $f^{i+1}_{a_i}(x)=x_{a_n}+\alpha_0+\dots+\alpha_i$ by induction on $1\leq i\leq n$. Since the context $z$ is negative, $\alpha$ is odd, so $f^\ell_{a_n}(x) = x_{a_n} + 1$ and thus $f^{2\ell}_{a_n}(x)=x_{a_n}$. Let $p$ be the period of $x$, and let $d$ be the period of $a_n$ in the limit cycle containing $x$, that is, $d$ is the least positive integer such that $f^k_{a_n}(x)=f^{k+d}_{a_n}(x)$ for all $k\in\N$. It is clear that $d\mid p$. Furthermore, we deduce from $f^\ell_{a_n}(x)= x_{a_n} + 1$ and $f^{2\ell}_{a_n}(x)=x_{a_n}$ that $d\nmid \ell$ and $d\mid 2\ell$. 
\end{proof}

\newpagexorapx
\section{Loop-to-cycle digraphs}

A strong component of a digraph with a single vertex may be \EM{trivial} (if it has no edges) or a \EM{loop} (if it has a single edge).
A \EM{loop-to-cycle} digraph is a digraph $G$ such that 
\begin{itemize}
\item
the terminal strong components are either trivial or cycles of length at least $2$,
\item
the initial strong components are trivial or loops, and
\item
all remaining strong components are trivial.
\end{itemize}
The cycles of length at least two are called \EM{terminal cycles}. 

We will prove that it is computationally hard to decide whether such seemingly simple digraphs $G$ are UFP-dynamic. On the other hand, there is a simple necessary and sufficient condition characterizing when a network over such a digraph has a unique fixed point, which is why we focus on this class.

Let $G$ be a loop-to-cycle digraph. Let $I$ be the vertices that are not in the terminal cycles, and let $L$ be the vertices with a loop. Let $f\in F(G)$. Suppose that some configuration in $\B^I$ is fixed by $f$. Then, necessarily, $f_a(x)=x_a$ for all $a\in L$ and for all $x\in \{0,1\}^G$, so every $z\in \B^L$ is fixed by $f$. Furthermore, since $G[I\setminus L]$ is acyclic and all the in-neighbors of vertices in $I\setminus L$ are in $I$, by Lemma~\ref{lem:robert}, there exists a unique $z'\in \B^I$ extending $z$ which is fixed by $f$, and $z$ converges to $z'$. Since $z'$ is a context for all the terminal cycles, we call $z'$ a \EM{fixed global context} of $f$. There are thus $2^{|L|}$ fixed global contexts, and $x_I$ is a fixed global context for every periodic configuration $x$ of $f$. As an immediate application of Lemma~\ref{lem:acyclic-positive}, we have the following lemma.

\begin{lemmarep}\label{lem:global_context}
Let $G$ be a loop-to-cycle digraph. Let $I$ be the vertices that are not in the terminal cycles. Let $f\in F(G)$ and let $z$ be a fixed global context. 
\begin{itemize}
\item
If $z$ is an acyclic context for every terminal cycle, then $f$ has a fixed point $x$ such that $z$ converges  to $x$.
\item
If $z$ is an acyclic or positive context for every terminal cycle, and a positive context for at least one, then $f$ has at least two fixed points extending $z$.
\item
If $z$ is a negative context for at least one terminal cycle, say of length $\ell$, then $z$ converges to a periodic configuration whose period is $q\cdot d$ for some $q \ge 1$ and some $d$ such that $d\mid 2\ell$ and $d\nmid \ell$. In particular, $f$ has no fixed point extending~$z$. 
\end{itemize}
\end{lemmarep}
We then obtain the following characterization (the easy proof is not given since we will not use this characterization). A network $f\in F(G)$ has $x$ as a unique fixed point if, and only if: 
\begin{itemize}
\item
$x_I$ is a fixed global context, which is acyclic for every terminal cycle, and
\item
every fixed global context $z\neq x_I$ is a negative context for at least one terminal cycle. 
\end{itemize}
By examining simple examples, one observes that the simultaneous satisfaction of these two conditions depends in a delicate way on the structure of the paths from loops to terminal cycles. This suggests that deciding whether these conditions can hold may be challenging, a fact that we will confirm. 
\newpagexorapx
\section{Odd-systems}

An \EM{odd-system} is a couple $\S=(I,C)$ where $I$ is a finite set of variables and $C$ is a set of couples of non-empty subsets of $I$, called \EM{odd-constraints}. The \EM{size} of $\S$ is $|I|\cdot|C|$. 
Let $x\in\B^I$. An odd-constraint $(A,B)$ is \EM{satisfied} by $x$ if $x_A=\ZERO$ or $x_B$ is odd, and $x$ is a \EM{solution} of $\S$ if all its odd-constraints are satisfied by $x$. We denote by \EM{$\Sol(\S)$} the set of solutions of $\S$. Note that $\ZERO$ is always a solution, and we set $\EMM{\Sol^+(\S)} = \Sol(\S) \setminus \{ \ZERO \} $. An odd-constraint of the form $(I,\{a,b\})$ is called a \EM{diff-constraint} between $a$ and $b$; it is satisfied by $\ZERO$ and all $x\neq \ZERO$ such that $x_a\neq x_b$. Graphical notations for odd-systems are given in Figure \ref{fig:graphical}.

\begin{figure}
\resizebox{\textwidth}{!}{
\begin{tikzpicture}
\begin{scope}
    \node[draw,rounded corners=4pt,inner sep=3pt] (node) at (1.75,0){\small $\ZERO\lor$odd};
    \node (0') at (0,2){\small $\bullet$};
    \node (1') at (0.5,2){\small $\bullet$};
    \node (2') at (1,2){\small $\bullet$};
    \node at (1.5,2){\small $\cdots$};
    \node at (1.75,3.1){$I$};
    \node at (2,2){\small $\bullet$};
    \node (3') at (2.5,2){\small $\bullet$};
    \node (4') at (3,2){\small $\bullet$};
    \node (5') at (3.5,2){\small $\bullet$};
    
    \node at (0,2.4){$a_1$};
    \node at (0.5,2.4){$a_2$};
    
    \node at (2.5,2.4){$b_1$};
    \node at (3,2.4){$b_2$};
    \node at (3.5,2.4){$b_3$};
    
    \draw[rounded corners=10pt,black!20] (-0.5,1.5) rectangle (4,2.8);

    \path
    (0') edge (node)
    (1') edge (node)
    (node) edge[decorate,decoration={snake,amplitude=0.2mm,segment length=2mm}] (3')
    (node) edge[decorate,decoration={snake,amplitude=0.2mm,segment length=2mm}] (4')
    (node) edge[decorate,decoration={snake,amplitude=0.2mm,segment length=2mm}] (5');
    
    ;
\end{scope}
\draw (4.5,-0.5) to (4.5,3.5);

\begin{scope}[xshift=5.5cm]
    \node[draw,rounded corners=4pt,inner sep=3pt] (node) at (1.75,0){\small $\ZERO\lor$odd};
    \node (0') at (0,2){\small $\bullet$};
    \node (1') at (0.5,2){\small $\bullet$};
    \node (2') at (1,2){\small $\bullet$};
    \node at (1.5,2){\small $\cdots$};
    \node at (1.75,3.1){$I$};
    \node (s) at (2,2){\small $\bullet$};
    \node (3') at (2.5,2){\small $\bullet$};
    \node (4') at (3,2){\small $\bullet$};
    \node (5') at (3.5,2){\small $\bullet$};
    
    \node at (0,2.4){};
    \node at (0.5,2.4){$a$};
    
    \node at (2.5,2.4){};
    \node at (3,2.4){$b$};
    \node at (3.5,2.4){};
    
    \draw[rounded corners=10pt,black!20] (-0.5,1.5) rectangle (4,2.8);

    \path
    (0') edge[bend right=3] (node)
    (2') edge (node)
    (1') edge[bend left=5] (node)
    (3') edge (node)
    (4') edge[bend right=5] (node)
    (5') edge[bend left=3] (node)
    (s) edge (node)
    (node) edge[bend left=5,decorate,decoration={snake,amplitude=0.2mm,segment length=2mm}] (1')
    (node) edge[bend right=5,decorate,decoration={snake,amplitude=0.2mm,segment length=2mm}] (4');

    ;
\end{scope}
\begin{scope}[xshift=10.5cm]
    \node (=) at (-0.8,1.2) {$=$};

    \node[draw,rounded corners=4pt,inner sep=3pt] (node) at (1.75,0){\small $\neq$};
    \node (0') at (0,2){\small $\bullet$};
    \node (1') at (0.5,2){\small $\bullet$};
    \node (2') at (1,2){\small $\bullet$};
    \node at (1.5,2){\small $\cdots$};
    \node at (1.75,3.1){$I$};
    \node at (2,2){\small $\bullet$};
    \node (3') at (2.5,2){\small $\bullet$};
    \node (4') at (3,2){\small $\bullet$};
    \node (5') at (3.5,2){\small $\bullet$};
    
    \node at (0.5,2.4){$a$};
    \node at (3,2.4){$b$};
    
    \draw[rounded corners=10pt,black!20] (-0.5,1.5) rectangle (4,2.8);
    
    \path%
    (1') edge (node)
    (4') edge (node)
    ;
\end{scope}

\end{tikzpicture}
}
\caption{
\textit{Left:} graphical notation for an odd-constraint $(A,B)$ with $A=\{a_1,a_2\}$ and $B=\{b_1,b_2,b_3\}$. \textit{Right:} simplification of the graphical notation for a diff-constraint between $a$ and $b$, which is the odd-constraint $(I,\{a,b\})$. 
}
\label{fig:graphical}
\end{figure}
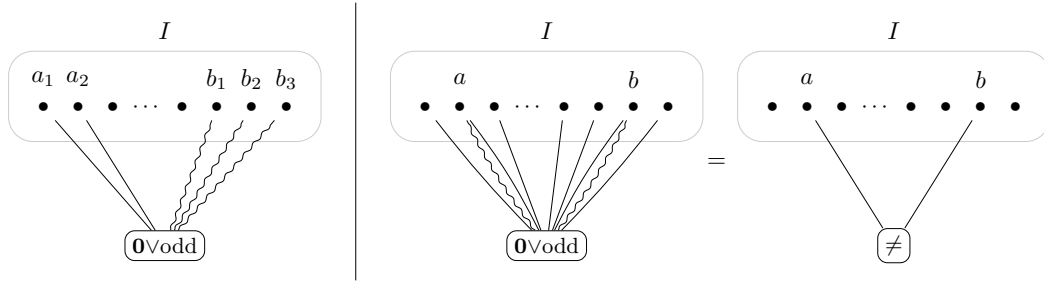

\subsection{Encoding circuits with odd-systems}

A \EM{circuit} is a Boolean network $f$ with the following properties. Its interaction graph $G$ is obtained from a non-empty  acyclic digraph, with maximum in-degree at most $2$, by adding a loop on each source. Vertices with a loop are the \EM{inputs} of the circuit. Non-input vertices are called \EM{gates}. Gates with a unique in-neighbor are \EM{not-gates}, and those with two in-neighbors are \EM{and-gates}. For every input $a$, we have $f_a(x)=x_a$. For every not-gate $a$ we have $f_a(x)=\neg x_b$, where $b$ is the in-neighbor of $a$, and for every and-gate $c$ we have $f_c(x)=x_a\land x_b$, where $a,b$ are the two in-neighbors of $c$. 
For technical reasons, we further require that $f$ contains at least one not-gate (possibly immediately preceded by another). This yields the property that $f$ does not fix $\ZERO$. 
Note that $f$ is completely defined by its interaction graph $G$. 

Let $V$ be the set of components of the circuit $f$, and let $L$ be the set of inputs of $f$. The \EM{size} of $f$ is $|V|$. An \EM{input configuration} $z$ is a configuration on $L$; it is obviously fixed by $f$. Furthermore, since $G\setminus L$ is acyclic, by Lemma~\ref{lem:robert}, $f$ has exactly one fixed point $x$ extending $z$, and $z$ converges to $x$; thus $f$ has exactly $2^{|L|}$ fixed points. Given a specified subset $J\subseteq V$ called \EM{output}, $f$ \EM{computes} the map $h$ from $\B^L$ to $\B^J$ defined by $h(z)=x_J$, where $x$ is the fixed point of $f$ extending $z$. Equivalently, $h(x_L) = x_J$ for every fixed $x\in\Fixe(f)$; all the information concerning computations is thus contained in $\Fixe(f)$. An example is given in Figure \ref{fig:circuit} (left). Note that for any map $h$ from $\{0,1\}^L$ to $\{0,1\}^J$ there exists some circuit $f$ which computes $h$.

This definition deviates slightly from the textbook definition of circuits (e.g., in \cite{P94}), but the correspondence between the object we describe and that referred to as a circuit there is straightforward. In particular, producing a ``Boolean network-style circuit'' from a ``textbook circuit'' or vice-versa is feasible in linear time, and for a given input configuration, the output computed by the former is precisely the same as that computed by the latter.  

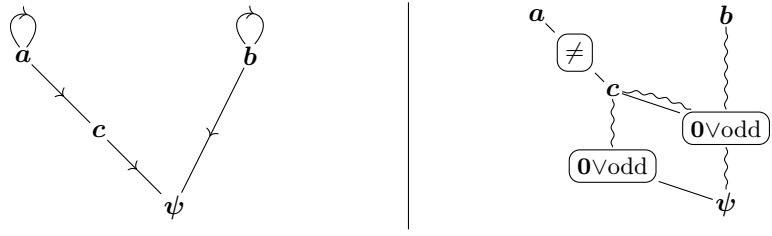
\begin{figure}
\[
\begin{array}{cccc}
\begin{array}{c}
\begin{tikzpicture}[every node/.style={outer sep=1,inner sep=0}]
\node (a) at (0,2){$\boldsymbol{a}$};
\node (b) at (3,2){$\boldsymbol{b}$};
\node (c) at (1,1){$\boldsymbol{c}$};
\node (psi) at (2,0){$\boldsymbol{\psi}$};

\draw[decoration={markings,mark=at position 0.55 with {\arrow{>}}},postaction={decorate}]  
(a.112) .. controls (-0.5,+2.7) and (+0.5,+2.7) .. (a.68);

\draw[decoration={markings,mark=at position 0.55 with {\arrow{>}}},postaction={decorate}]  
(b.112) .. controls (2.5,+2.7) and (+3.5,+2.7) .. (b.68);

\path[decoration={markings,mark=at position 0.55 with {\arrow{>}}}]
(a) edge[postaction={decorate}] (c)
(b) edge[postaction={decorate}] (psi)
(c) edge[postaction={decorate}] (psi)
;
\end{tikzpicture}
\end{array}
\qquad
\qquad
\vline
\qquad
\qquad

\begin{array}{c}
\begin{tikzpicture}[every node/.style={outer sep=1,inner sep=0}]
\node (a) at (0.5,2.5){$\boldsymbol{a}$};
\node (b) at (3,2.5){$\boldsymbol{b}$};
\node (c) at (1.5,1.5){$\boldsymbol{c}$};
\node (psi) at (3,0){$\boldsymbol{\psi}$};

\node[draw,rounded corners=4pt,inner sep=3pt] (ac) at (1,2){\small $\neq$};
\node[draw,rounded corners=4pt,inner sep=3pt] (psic) at (1.5,0.5){\small $\ZERO\lor$odd};
\node[draw,rounded corners=4pt,inner sep=3pt] (ccbpsi) at (3,1){\small $\ZERO\lor$odd};

\path
(a) edge (ac)
(ac) edge (c)
(psi) edge (psic)
(psic) edge[decorate,decoration={snake,amplitude=0.2mm,segment length=2mm}] %
(c)
(c) edge (ccbpsi)
(ccbpsi) edge[decorate,decoration={snake,amplitude=0.2mm,segment length=2mm}, bend right=10] (c)
(ccbpsi) edge[decorate,decoration={snake,amplitude=0.2mm,segment length=2mm}] (b)
(ccbpsi) edge[decorate,decoration={snake,amplitude=0.2mm,segment length=2mm}] (psi)
;
\end{tikzpicture}
\end{array}
\end{array}
\]
{\caption{\label{fig:circuit}
\textit{Left:} Interaction graph of a circuit $f$. Thus, $x$ is a fixed point of $f$ if and only if $x_c=\neg x_a$ and $x_\psi=x_c\land x_b$. Hence $x_\psi=\neg x_a\land x_b$ for the four fixed points $x$ of $f$. \textit{Right:} The tight encoding $\S_f$ of the circuit $f$.} 
}
\end{figure}

An odd-system $\S$ \EM{encodes} $f$ if $V$ is included in the variable set of $\S$ and 
\[
\Sol^{\textrm{\tiny $+$}}(\S)_V=\Fixe(f).
\]

Let $f$ be a circuit with set of components $V$, and let $\S_f = (I, C)$ be the odd-system defined as follows. The variable set is $I = V$. The constraints are as follows: for each not-gate $u$, with in-neighbor $v$, $\S_f$ contains a diff-constraint between $u$ and $v$; for each and-gate~$w$, with in-neighbors $u,v$, $\S_f$ contains the two odd-constraints $( \{ w \}, \{ u \} )$ and $( \{ u \}, \{ u,v,w \} )$. Figure~\ref{fig:circuit} (right) gives an example.

\begin{lemmarep}\label{lem:encoding_circuit}
The odd-system $\S_f$ encodes the circuit $f$.
\end{lemmarep}

\begin{proof}
Since $f$ contains a not-gate, $\ZERO$ is neither a fixed point of $f$ nor a solution of $\S_f$.

It is clear that $x \in \{0,1\}^I$ satisfies the diff-constraint between the not-gate $b$ and its in-neighbor $a$ if and only if $x_b = \neg x_a$. Similarly, the two odd-constraints for the and-gate $c$ with inputs $a$ and $b$ are satisfied by $x \in \{0,1\}^I$ if and only if $x_c = x_a \land x_b$, as seen from the table below.

\begin{tabular}{c|c|c}
     $x_a x_b x_c$ & $( \{ c \}, \{ a \} )$ & $( \{ a \}, \{ a,b,c \} )$\\
     \hline
     \textcolor{black}{$000$} & T & T \\
     \textcolor{red}{$001$} & \textcolor{red}{F} & T \\
     \textcolor{black}{$010$} & T & T \\
     \textcolor{red}{$011$} & \textcolor{red}{F} & T \\
     \textcolor{black}{$100$} & T & T \\
     \textcolor{red}{$101$} & T & \textcolor{red}{F} \\
     \textcolor{red}{$110$} & T & \textcolor{red}{F} \\
     \textcolor{black}{$111$} & T & T 
\end{tabular}

Thus $\Sol^+( \S_f ) = \Sol( \S_f ) = \Fixe(f)$, and hence $\S_f$ encodes $f$. 
\end{proof}

In view of Lemma~\ref{lem:encoding_circuit} and the fact that $I(S_f)=V$,
we shall refer to $\S_f$ as the \EM{tight encoding} of the circuit $f$. We note that $\S_f$ can be obtained in linear time from $f$.

\subsection{Encoding Boolean formulae with odd-systems}

Every Boolean formula $\psi$ with variable set $\Lambda$ can be naturally encoded by a circuit $f$, where the input set is $\Lambda$, the output is unique and denoted $\psi$, and the corresponding computed map is the Boolean function that gives the truth table of $\psi$. In other words, the set of satisfying assignments of $\psi$ is $\{x_\Lambda\mid x\in \Fixe(f)\textrm{ and }x_\psi=1\}$. For instance, the circuit of Figure~\ref{fig:circuit} (left) encodes the formula $\psi=\neg a\land b$. Let $\S$ be an odd-system encoding $f$. The set of satisfying assignments of $\psi$ is also $\{x_\Lambda\mid x\in \Sol^+(\S)\textrm{ and }x_\psi=1\}$.
 
\newpagexorapx
\section{Odd-system vs digraph problem}

The protagonist of this section is the following problem:

\begin{quote}\decisionpb
	{Odd-System-vs-Digraph (\OSvsD)}
	{an odd-system $\S$ and a digraph $H$ on the variable set~of~$\S$.}
	{does there exist $h\in F_\ZERO(H)$ such that $\Fixe^+(h)\cap \Sol^+(\S)=\emptyset$?}
\end{quote}

In this section we prove that {\OSvsD} reduces to \UFP. 
This is key to the remainder of the paper since the computational hardness of {\OSvsD} is the central ingredient of our reductions proving hardness of \UFP.
To illustrate this, let $f$ be a circuit encoding a Boolean formula $\psi$, and let $\S_f$ be the tight encoding of $f$ (which can be constructed in linear time). Let $H$ be the digraph defined as follows: the vertex set $I$ consists of the variables of $\S$; there is a loop on every vertex except $\psi$, and no other arcs. For every $h\in F_\ZERO(H)$, we have $\Fixe(h)=\{x\in\B^I\mid x_\psi=0\}$. Since $\S$ encodes $f$, we obtain 
\[
\textrm{$\psi$ is a tautology}\iff \Sol^+(\S)=\{x\in \Sol^+(\S)\mid x_\psi=1\}
\iff \Fixe^+(h)\cap \Sol^+(\S)=\emptyset.
\]
Hence {\sc Tautology} reduces to {\OSvsD}, the latter thus being $\coNP$-hard. We will significantly improve this result in the next sections. 

Given an instance $\S,H$, the reduction consists in constructing a digraph $G$ obtained from $H$ by  ``plugging in'' copies of a gadget coding the odd-constraints of $\S$. Let us describe this precisely. 

Given positive integers $n,m$, the \EM{$(n,m)$-gadget} is the digraph obtained from the disjoint union of a path of length $n$ and a cycle of length $m+1$ by adding a loop on the first vertex of the path, and an arc from the last vertex of the path to a vertex of the cycle; see Figure~\ref{fig:gadget}$(a)$. Hence the $(n,m)$-gadget has $n$ linear vertices that do not belong to the cycle (green squares in the figure) and $m$ linear vertices that belong to the cycle (red triangles in the figure). Here and throughout this paper, a vertex is called linear if and only if it has in- and out-degree one.

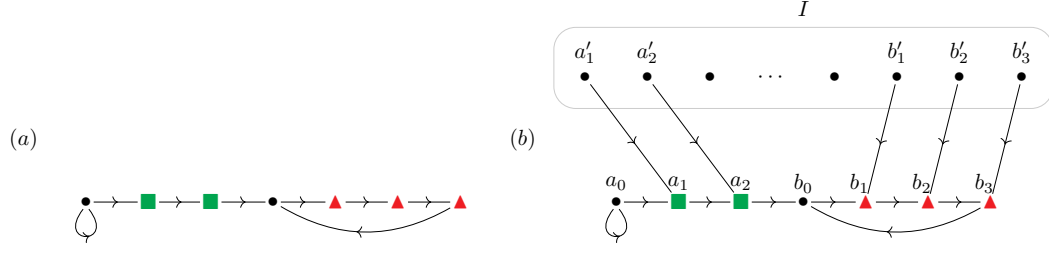
\begin{figure}
\resizebox{\textwidth}{!}{
\begin{tikzpicture}[every node/.style={outer sep=1,inner sep=0}]
    \begin{scope}
    \node (a) at (-1,1) {$(a)$};
    
    \node (0) at (0,0){$\bullet$};
    \node (1) at (1,0){\textcolor{Green}{$\blacksquare$}};
    \node (2) at (2,0){\textcolor{Green}{$\blacksquare$}};
    \node (3) at (3,0){$\bullet$};
    \node (4) at (4,0){\textcolor{Red}{$\blacktriangle$}};
    \node (5) at (5,0){\textcolor{Red}{$\blacktriangle$}};
    \node (6) at (6,0){\textcolor{Red}{$\blacktriangle$}};

    \draw[decoration={markings,mark=at position 0.55 with {\arrow{>}}},postaction={decorate}]  
    (0.-112) .. controls (-0.5,-0.7) and (0.5,-0.7) .. (0.-68);
    
    \draw[decoration={markings,mark=at position 0.55 with {\arrow{>}}},postaction={decorate},color=white]  
    (0.+112) .. controls (-0.5,+0.7) and (0.5,+0.7) .. (0.+68);
    
    \path[decoration={markings,mark=at position 0.55 with {\arrow{>}}}]
    (0) edge[postaction={decorate}] (1)
    (1) edge[postaction={decorate}] (2)
    (2) edge[postaction={decorate}] (3)
    (3) edge[postaction={decorate}] (4)
    (4) edge[postaction={decorate}] (5)
    (5) edge[postaction={decorate}] (6)
    (6) edge[postaction={decorate},bend left=30] (3)
    ;
\end{scope}
\begin{scope}[xshift=8cm]
    \node (b) at (-1,1) {$(b)$};

    \node (0) at (0.5,0){$\bullet$};
    \node (1) at (1.5,0){\textcolor{Green}{$\blacksquare$}};
    \node (2) at (2.5,0){\textcolor{Green}{$\blacksquare$}};
    \node (3) at (3.5,0){$\bullet$};
    \node (4) at (4.5,0){\textcolor{Red}{$\blacktriangle$}};
    \node (5) at (5.5,0){\textcolor{Red}{$\blacktriangle$}};
    \node (6) at (6.5,0){\textcolor{Red}{$\blacktriangle$}};
    \node (0') at (0,2){$\bullet$};
    \node (1') at (1,2){$\bullet$};
    \node (2') at (2,2){$\bullet$};
    \node at (3,2){$\cdots$};
    \node at (3.5,3.1){$I$};
    \node at (4,2){$\bullet$};
    \node (3') at (5,2){$\bullet$};
    \node (4') at (6,2){$\bullet$};
    \node (5') at (7,2){$\bullet$};
    
    \node at (0,2.4){$a'_1$};
    \node at (1,2.4){$a'_2$};
    
    \node at (5,2.4){$b'_1$};
    \node at (6,2.4){$b'_2$};
    \node at (7,2.4){$b'_3$};
    
    \node at (0.5,0.3){$a_0$};
    \node at (1.5,0.3){$a_1$};
    \node at (2.5,0.3){$a_2$};
    \node at (3.5,0.3){$b_0$};
    \node at (4.4,0.3){$b_1$};
    \node at (5.4,0.3){$b_2$};
    \node at (6.4,0.3){$b_3$};
    
    \draw[rounded corners=10pt,black!20] (-0.5,1.5) rectangle (7.5,2.8);
    
    \draw[decoration={markings,mark=at position 0.55 with {\arrow{>}}},postaction={decorate}]  
    (0.-112) .. controls (0,-0.7) and (1,-0.7) .. (0.-68);
    
    \path[decoration={markings,mark=at position 0.55 with {\arrow{>}}}]
    (0) edge[postaction={decorate}] (1)
    (1) edge[postaction={decorate}] (2)
    (2) edge[postaction={decorate}] (3)
    (3) edge[postaction={decorate}] (4)
    (4) edge[postaction={decorate}] (5)
    (5) edge[postaction={decorate}] (6)
    (6) edge[postaction={decorate},bend left=30] (3)
    (0') edge[postaction={decorate}] (1)
    (1') edge[postaction={decorate}] (2)
    (3') edge[postaction={decorate}] (4)
    (4') edge[postaction={decorate}] (5)
    (5') edge[postaction={decorate}] (6)
    ;    
\end{scope}
\end{tikzpicture}
}
\caption{$(a)$ The $(2,3)$-gadget. $(b)$ The $(A,B)$-extension with $A=\{a'_1,a'_2\}$ and $B=\{b'_1,b'_2,b'_3\}$ (labels are those used in Lemma~\ref{lem:AB-extension} in the appendix).}
\label{fig:gadget}
\end{figure}

We now describe the ``plugging'' operation. Let $H$ be a digraph and $I\subseteq V(H)$. Given non-empty $A,B\subseteq I$, the \EM{$(A,B)$-extension} of $H$ from $I$ is the digraph obtained from $H$ by adding 
\begin{itemize}
\item 
a copy of the $(|A|,|B|)$-gadget,
\item 
a matching from $A$ to the linear vertices that do not belong to the cycle, 
\item 
a matching from $B$ to the linear vertices that belong to the cycle. 
\end{itemize}
See Figure~\ref{fig:gadget}$(b)$. Let $\S=(I,C)$ be an odd-system, and let $H$ be a digraph on $I$. The \EM{$\S$-extension} of $H$ is the digraph $G$ obtained from $H$ by applying successively, for each $(A,B)\in C$, an $(A,B)$-extension of $H$ from $I$. The order in which the odd-constraints are processed is immaterial to the graph $G$ obtained. Hence $G$ is constructed in polynomial time.  Furthermore, if $H$ is a loop-to-cycle digraph, then $G$ is still a loop-to-cycle digraph, and $\Delta^-(G)\leq\max(2,\Delta^-(H))$ because all the vertices that belong to the added gadgets have in-degree at most $2$, and since there is no arc from these vertices to vertices of $H$.

\begin{restatable}{theorem}{oddexttheorem}\label{thm:odd-extension}
Let $\S,H$ be an instance of {\OSvsD}. Let $G$ be the $\S$-extension of $H$. We have 
\[
\textrm{$\S,H$ is a positive instance of {\OSvsD}}\iff \text{$G$ is UFP-dynamic}.
\]
\end{restatable}

For the proof we need some definitions. Let $G$ be a digraph and $I\subseteq V(G)$. We denote by \EM{$F_\ZERO(G,I)$} the set of $f\in F_\ZERO(G)$ such that $f$ has no fixed point $x\neq \ZERO$ with $x_I=\ZERO$; in other words, $\ZERO$ is the unique fixed point of $f$ extending $\ZERO_I$. Thus, for any $G$ and $I \subseteq V(G)$, $UFP_\ZERO(G) \subseteq F_\ZERO(G,I)\subseteq F_\ZERO(G)$.

Suppose that all the in-neighbors of vertices in $I$ are in $I$, and let $H=G[I]$. The \EM{restriction} of $f\in F(G)$ on $H$ is the network $h\in F(H)$ defined by $h(x_I)=f(x)_I$ for all $x\in\B^{G}$; this is well-defined since all the in-neighbors of vertices in $I$ are in $I$. Given $h\in F(H)$, we say that $f\in F(G)$ \EM{extends} $h$ if $h$ is the restriction of $f$ on $H$.

\begin{proofsketch}
The major part of the proof is devoted to proving that the $(A,B)$-extension of $H$ adds the constraint ``$x_A=\ZERO\textrm{ or }x_B\textrm{ is odd}$'' to the set of fixed points of the Boolean network extension of $h \in F_\ZERO( H,I )$ that minimizes the number of fixed points. In other words, the $(A,B)$-extension of digraphs is an application of the odd-constraint $(A,B)$ to the minimal set of fixed points. More precisely, let $h\in F_\ZERO(H,I)$ and $G$ be the $(A,B)$-extension of $H$ from $I$. Then there exists $f\in F_\ZERO(G,I)$ extending $h$, and if $f$ minimizes $|\Fixe(f)_I|$ under these conditions then 
\[
    \Fixe(f)_I=\{x\in\Fixe(h)_I\mid x_A=\ZERO\textrm{ or }x_B\textrm{ is odd}\}.
\]
Intuitively, this allows us to treat  all the odd-constraints of $\S$ as applying to the set of fixed points.
More precisely, let $G$ be the $\S$-extension of $H$ and $h\in F_\ZERO(H)$. Then there exists $f\in F_\ZERO(G,I)$ extending $h$, and if $f$ minimizes $|\Fixe(f)_I|$ under these conditions then 
\[
    \Fixe^+(f)_I=\Fixe^+(h)\cap \Sol^+(\S).
\]
We can then conclude that there exists $h \in F_\ZERO( H )$ with $\Fixe^+(h)\cap \Sol^+(\S) = \emptyset$ if and only if there exists $f \in UFP_\ZERO( G )$.
\end{proofsketch}

\begin{toappendix}

\begin{lemmarep}\label{lem:AB-extension}
Let $H$ be a digraph and $I\subseteq V(H)$. Let $A,B\subseteq I$ non-empty, and let $G$ be the $(A,B)$-extension of $H$ from $I$. Let $h\in F_\ZERO(H,I)$. There exists $f\in F_\ZERO(G,I)$ extending $h$, and if $f$ minimizes $|\Fixe(f)_I|$ under these conditions then 
\[
\Fixe(f)_I=\{x\in\Fixe(h)_I\mid x_A=\ZERO\textrm{ or }x_B\textrm{ is odd}\}.
\]
\end{lemmarep}

\begin{proof}
Let $n=|A|$ and $m=|B|$. Let us label the $(n,m)$-odd-gadget associated with the $(A,B)$-extension as follows. The loop is denoted $a_0$, and the linear vertices that do not belong to the cycle are labelled from $a_1$ to $a_n$, following the order of the path containing them. The vertices of the cycle are denoted from $b_0$ to $b_m$ following the order, the starting vertex $b_0$ being the non-linear vertex of the cycle. Figure~\ref{fig:gadget}$(b)$ gives an illustration for $n=2$ and $m=3$. For $1\leq i\leq n$, we denote by $a'_i$ the in-neighbor of $a_i$ in $A$, and for $1\leq i\leq m$, we denote by $b'_i$ the in-neighbor of $b_i$ in $B$. The proof is divided into three claims, and the lemma is a direct consequence of the first and the third. 

\begin{claim}
There exists $f\in F_\ZERO(G,I)$ extending $h$ such that 
\[
\Fixe(f)_I\subseteq \{x\in\Fixe(h)_I\mid x_A=\ZERO\textrm{ or }x_B\textrm{ is odd}\}.
\]
\end{claim}

\begin{claimproof}
Let $f\in F(G)$ be the extension of $h$ such that
\[
\begin{array}{l}
f_{a_0}(x)=x_{a_0}\\
f_{a_i}=x_{a_{i-1}}\lor x_{a'_i}\quad (1\leq i\leq n)\\
f_{b_0}(x)=\neg x_{b_m}\land x_{a_n}\\
f_{b_i}(x)=x_{b_{i-1}}+x_{b'_i}\quad (1\leq i\leq m).
\end{array}
\]

Let us prove that $f\in F_\ZERO(G,I)$. Since $h(\ZERO)=\ZERO$ and by construction of $f$ as above we have $f(\ZERO)=\ZERO$. Suppose $x\in\Fixe(f)$ and $x_I=\ZERO$. If $x_{a_n}=1$ then $x_{b_0}=\neg x_{b_m}$ and $x_{b_i}=x_{b_{i-1}}$ for $1\leq i\leq m$, thus $x_{b_0}=x_{b_m}$, a contradiction. So $x_{a_n}=0$ and we deduce that $x_{a_i}=0$ for $0\leq i<n$. We also deduce that $x_{b_0}=0$ thus $x_{b_i}=0$ for $1\leq i\leq m$. Thus $x=\ZERO$. This proves that $f\in F_\ZERO(G,I)$.

We now prove the inclusion. Let $x\in\Fixe(f)$. Then $x$ extends a fixed point of $h$, so we only have to prove that $x_A=\ZERO$ or $x_B$ is odd. Suppose that $x_A\neq\ZERO$. We deduce that $x_{a_n}=1$. Thus $x_{b_0}=\neg x_{b_m}$ and $x_{b_i}=x_{b_{i-1}}+x_{b'_i}$ for $1\leq i\leq m$, which implies $x_{b_m}=x_{b_0}+x_{b'_1}+\cdots +x_{b'_m}$. Since $x_{b_0}=\neg x_{b_m}$ we deduce that $x_{b'_1}+\cdots +x_{b'_m}=1$, thus $x_B$ is odd.
\end{claimproof}

\begin{claim} \label{clm:AB-extension-2}
If $f\in F_\ZERO(G,I)$ then 
\[
\begin{array}{l}
f_{a_0}(x)=x_{a_0}\\
f_{a_i}(x)=x_{a_{i-1}}\textrm{ if }x_{a'_i}=0\qquad (1\leq i\leq n)\\
f_{b_0}(x)=\neg x_{b_m}\land x_{a_n}\\
f_{b_i}(x)=x_{b_{i-1}}\textrm{ if }x_{b'_i}=0\qquad (1\leq i\leq m).
\end{array}
\]
\end{claim}

\begin{claimproof}
Let $f\in F_\ZERO(G,I)$. Since $f(\ZERO)=\ZERO$, we obviously have $f_{a_0}(x)=x_{a_0}$. For $0\leq i\leq n$, let $x^i$ be the sum of $e_{a_j}$ for $0\leq j\leq i$. Since $x^0$ is not a fixed point because $f\in F_\ZERO(G,I)$ and $a_0\notin I$, we have $f_{a_1}( x^0 ) = f_{a_1}( e_{a_0} ) = 1$. By induction on $i$, since $x^{i-1}$ is not a fixed point, we have $f_{a_i}(e_{a_{i-1}})=f_{a_i}(x^{i-1})=1$, and thus $f_{a_i}(x)=x_{a_{i-1}}$ if $x_{a'_i}=0$. Since $x^n$ is not a fixed point, we have $f_{b_0}(e_{a_n})=f_{b_0}(x^n)=1$. For $1\leq i\leq m$, let $y^i$ be the sum of $e_{b_j}$ for $0\leq j\leq i$. Since $x^n+y^i$ is not a fixed point, we have $f_{b_i}(e_{b_{i-1}})=f_{b_i}(x^n+y^{i-1})=1$, and thus $f_{b_i}(x)=x_{b_{i-1}}$ if $x_{b'_i}=0$. Since $x^n+y^m$ is not a fixed point, we have $f_{b_0}(e_{a_n}+e_{b_m})=0$. Since $y^m$ is not a fixed point, we have $f_{b_0}(e_{b_m})=f_{b_0}(y^m)=0$. So we have proved that $f_{b_0}(e_{a_n})=1$, $f_{b_0}(e_{a_n}+e_{b_m})=0$ and $f_{b_0}(e_{b_m})=0$, (and we have that $f_{b_0}(\ZERO)=0$) which means that $f_{b_0}(x)=\neg x_{b_m}\land x_{a_n}$.
\end{claimproof}

\begin{claim}
Let $f\in F_\ZERO(G,I)$ extending $h$. We have 
\[
\{x\in\Fixe(h)_I\mid x_A=\ZERO\textrm{ or }x_B\textrm{ is odd}\}\subseteq \Fixe(f)_I.
\]
\end{claim}

\begin{claimproof}
Let $z\in \Fixe(h)$ such that $z_A=\ZERO$ or $z_B$ is odd, and let us prove that $f$ has a fixed point extending $z$. Since $z$ is fixed by $f$, by Lemma \ref{lem:robert}, there is a configuration $z'$ on $I\cup\{a_0,\dots,a_n\}$ extending $z$ which is fixed by $f$ and with $z'_{a_0}=0$. Since $z'$ is a context for the cycle $C = \{ b_0, \dots, b_m \}$ of the gadget, by Lemma \ref{lem:acyclic-positive}, it is sufficient to prove that $z'$ is not a negative context. Suppose that $z'$ is a functional context for $C$, let $\alpha=(\alpha_0,\dots,\alpha_m)$ be the corresponding signature, and let us prove that $z'$ is then necessarily a positive context. By Claim \ref{clm:AB-extension-2}, if $z'_{a_n}=0$, then $b_0$ is fixed to $0$ in the context $z'$, which is thus acyclic, a contradiction. Thus $z'_{a_n}=1$. We deduce from Claim \ref{clm:AB-extension-2} that $z_A\neq\ZERO$ and $\alpha_0=1$. Let $1\leq i\leq m$. By Claim \ref{clm:AB-extension-2}, if $z_{b'_i}=0$ then $\alpha_i=0$. Suppose that $z_{b'_i}=1$. If $\alpha_i=0$ then $f_{b_i}(e_{b'_i})=0$ and $f_{b_i}(e_{b_{i-1}}+e_{b'_i})=1$. By Claim \ref{clm:AB-extension-2}, $f_{b_i}(\ZERO)=0$ and $f_{b_i}(e_{b_{i-1}})=1$. It means that $f_{b_i}(x)=x_{b_{i-1}}$ for all configurations $x$ on $G$, so $f_{b_i}$ does not depend on $b'_i$, a contradiction. Thus $\alpha_i=1$. So $\alpha_i=z_{b'_i}$ in any case. Hence we have proved that $\alpha=(1,z_{b'_1},\dots,z_{b'_m})$. Since $z_B$ is odd, $\alpha$ is even. Hence $z'$ is indeed a positive context. 
\end{claimproof}
\end{proof}

\begin{lemmarep}\label{lem:odd-extension}
Let $\S$ be an odd-system with variable set $I$, $H$ be a digraph on $I$, and $G$ be the $\S$-extension of $H$. Let $h\in F_\ZERO(H)$. There exists $f\in F_\ZERO(G,I)$ extending $h$, and if $f$ minimizes $|\Fixe(f)_I|$ under these conditions then 
\[
\Fixe^+(f)_I=\Fixe^+(h)\cap \Sol^+(\S).
\]
\end{lemmarep}

\begin{proof} 
We proceed by induction on the number of odd-constraints in $\S$. If $\S$ has no odd-constraint then $G=H$, and $f=h$ has the desired properties since $\Sol(\S)=\B^I$. 

Let $C$ be the set of odd-constraints in $\S$, and suppose that $(A,B)\in C$. Let $\S'=(I,C\setminus \{(A,B)\})$ and $\S''=(I,\{(A,B)\})$. Thus $\Sol(\S)=\Sol(\S')\cap\Sol(\S'')$. Let $G'$ be the $\S'$-extension of $H$; hence $G$ is the $(A,B)$-extension of $G'$ from $I$. 

Let any $f\in F_\ZERO(G,I)$ extending $h$, and let $f'$ be its restriction on $G'$. Suppose that $f'$ has a fixed point $x'\neq \ZERO$ with $x'_I=\ZERO$. Let $U=V(G)\setminus V(G')$, and let $x\in\B^G$ extending  $x'$ with $x_U=\ZERO$. Since the in-neighbors of vertices in $U$ are in $U\cup I$, since $x_{U\cup I}=\ZERO$ and since $f(\ZERO)=\ZERO$, we deduce that $x\in\Fixe^+(f)$ and $x_I=x'_I=\ZERO$, a contradiction. Thus $f'\in F_\ZERO(G',I)$. Since $G$ is the $(A,B)$-extension of $G'$ from $I$, we deduce from Lemma \ref{lem:AB-extension} that 
\[
\Fixe(f')_I\cap\Sol(\S'')\subseteq \Fixe(f)_I. 
\]
Since $f'$ extends $h$, using the induction for the first inclusion, we obtain: 
\[
\Fixe^+(h)\cap \Sol^+(\S)
=\Fixe^+(h)\cap \Sol^+(\S')\cap\Sol^+(\S'')
\subseteq \Fixe^+(f')_I\cap \Sol^+(\S'')\subseteq \Fixe(f)_I.
\]
It remains to prove that these inclusions are equalities for some $f\in F_\ZERO(G,I)$ extending $h$. By induction, there exists $f'\in F_\ZERO(G',I)$ extending $h$ such that $\Fixe^+(f')_I=\Fixe^+(h)\cap \Sol^+(\S')$. By Lemma \ref{lem:AB-extension} there exists $f\in F_\ZERO(G,I)$ extending $f'$ such that $\Fixe(f)_I=\Fixe(f')_I\cap\Sol(\S'')$. 
We then deduce that $\Fixe^+(f)_I=\Fixe^+(h)\cap \Sol^+(\S)$. This completes the proof.
\end{proof}

\end{toappendix}

\begin{toappendix}
We are now ready to prove the main theorem for this section, which we first recall:

\oddexttheorem*

\begin{proof}[{\bf Proof of Theorem~\ref{thm:odd-extension}}]
Let $I$ be the vertex set of $H$.

Suppose that $G$ is UFP-dynamic. By Lemma~\ref{lem:UFP0}, there exists $f\in F(G)$ with $\ZERO$ as unique fixed point. Let $h$ be the restriction of $f$ on $H$. Then $h\in F_\ZERO(H)$, $f$ extends $h$, and $f\in F_\ZERO(G,I)$ ($\subseteq F_\ZERO(G)$). By Lemma~\ref{lem:odd-extension}, we have $\Fixe^+(h)\cap \Sol^+(\S)=\Fixe^+(f)_I=\emptyset$. Thus $\S,H$ is a positive instance of \OSvsD.

Conversely, suppose that there exists $h\in F_\ZERO(H)$ with $\Fixe^+(h)\cap \Sol^+(\S)=\emptyset$. By Lemma~\ref{lem:odd-extension}, there exists $f\in F_\ZERO(G,I)$ extending $h$ such that $\Fixe^+(f)_I=\Fixe^+(h)\cap \Sol^+(\S)=\emptyset$. It means that if $x$ is a fixed point of $f$ then $x_I=\ZERO$ and thus $x=\ZERO$ since $f\in F_\ZERO(G,I)$. So $\ZERO$ is the unique fixed point of $f$. Thus $G$ is UFP-dynamic.
\end{proof}
\end{toappendix}

\newpagexorapx
\section{Complexity lower bound for \texorpdfstring{$\dUFP$}{d-UFP}}\label{sec:dUFP}

\begin{quote}\decisionpb
	{Quantified satisfiability with 2 alternating quantifiers ~(\qsat)}
	{a Boolean formula $\psi$ with variable set $\Lambda$ and $E\subseteq \Lambda$.}
	{does there exist $z\in\B^E$ such that $\psi$ is satisfied by all $x\in\B^\Lambda$ extending $z$?}
\end{quote}

\begin{quote}
\decisionpb
	{{\dUFP} restricted to loop-to-cycle digraphs (\dLCUFP)}
	{a loop-to-cycle digraph $G$ with $\Delta^-(G)\leq d$.}
	{is $G$ UFP-dynamic?}
\end{quote}

We prove below that {\qsat} reduces to {\dUFP} for all $d\geq 2$. Since {\qsat} is $\NPoNP$-complete \cite[Theorem 17.10]{P94}, this proves that {\dUFP} is $\NPoNP$-hard for all $d\geq 2$. Actually, the reduction even shows that {\dLCUFP} is $\NPoNP$-hard for all $d\geq 2$. Combined with Corollary~\ref{cor:in-np-or-nponp} this yields the second half of Theorem~\ref{thm:UFP}: $\NPoNP$-completeness of $d$-\UFP.

\begin{theoremrep}\label{thm:qsat}
Let $\psi,E$ be an instance of \qsat. We can construct in polynomial time a loop-to-cycle digraph $G_{\psi,E}$, with maximum in-degree $2$, such that:
\[
\text{$\psi,E$ is a positive instance of \qsat} \iff \text{$G_{\psi,E}$ is UFP-dynamic} .
\] 
\end{theoremrep}

\begin{proofsketch}
The reduction is through \OSvsD. The main idea of the proof is to construct three objects: a digraph $H$, an odd-system $\S$, and $G = G_{\psi, E}$ the $\S$-extension of $H$. The construction is as follows. Let $f^\psi$ be the circuit encoding $\psi$ and let $\S_\psi = \S_{ f^\psi }$ be the tight encoding of $f^\psi$.  The vertices of $H$ are the variables of $\S_\psi$, which include the inputs $\Lambda$ and the output $\psi$, and an additional vertex $v$. Figure \ref{fig:combined-ufp-complexity-construction} (left) is an illustration of the graph $H$ and of the single constraint in $\S$ which is not also in $\S_\psi$. 
The proof then shows that $\S$ encodes $\psi$, and that $\psi,E$ is a positive instance of \qsat{} if and only if $G_{\psi,E}$ is UFP-dynamic.
\end{proofsketch}

\begin{toappendix}
    \begin{figure}[!ht]
    \centering
    \includegraphics[page=8]{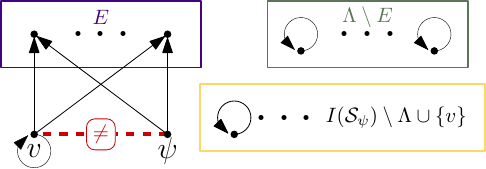}
    \caption{Graph $H$ from our construction. Vertices in the yellow box (bottom right) are in the graph $H$ but are relevant to the proof only through the constraints of $\S_\psi$, which are not shown.}
    \label{fig:d-ufp-complexity-construction}
\end{figure}
\end{toappendix}

\begin{figure}[!ht]
    \centering
    \includegraphics[page=9]{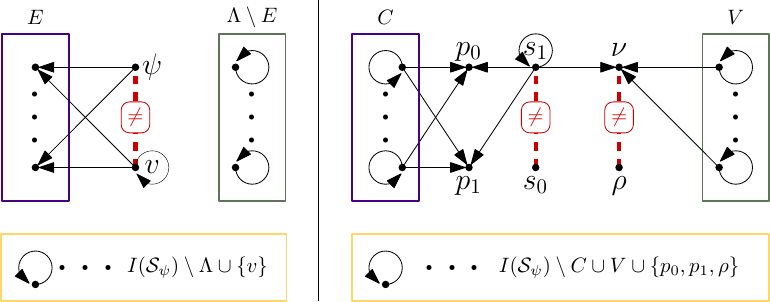}
    \caption{\textit{Left:} Graph $H$ from our construction for Theorem~\ref{thm:qsat}. 
    \textit{Right:} Graph $H$ from our construction in our construction for Theorem:~\ref{thm:succinct-4cnf-iff-ufp}.
    \textit{Both:} Vertices in the yellow box (at the bottom) are in the graph $H$ but are relevant to the proofs only through the constraints of $\S_\psi$, which are not shown.}
    \label{fig:combined-ufp-complexity-construction}
\end{figure}

\begin{proof}
Let $\Lambda$ be the variable set of $\psi$, and let $f^\psi$ be the circuit encoding $\psi$. Let $\S_\psi = \S_{ f^\psi }$ be the tight encoding of $f^\psi$. Let $H$ be the digraph defined as follows. The vertices are the variables of $\S_\psi$, which include the inputs $\Lambda$ and the output $\psi$, and an additional vertex $v$. The reader may find Figure \ref{fig:d-ufp-complexity-construction} helpful in following the construction and proof. The arcs are:
\begin{itemize}
\item 
a loop on each vertex not in $E\cup\{\psi\}$,
\item 
an arc from each vertex in $\{\psi,v\}$ to each vertex in $E$.
\end{itemize}

Let $\S$ be the odd-system defined as follows. The set of variables $I$ of $\S$ is the vertex set of $H$, and the constraint set consists of:
\begin{itemize}
\item 
all the constraints of $\S_\psi$, 
\item 
a diff-constraint between $\psi$ and $v$.
\end{itemize}

\begin{claim}\label{clm:d-ufp-s-encodes-psi}
     $\S$ encodes $\psi$, that is $\PosSol(\S)_{I(\S_\psi)}=\Fixe(f^\psi)$.
\end{claim}
\begin{claimproof}
    First, suppose $x \in \PosSol(\S)_{I(\S_\psi)}$. Clearly, $x \in \Sol(\S_\psi)_{I(\S_\psi)}=\Sol(\S_\psi)$ since the restriction of $x$ to the variable set of $\S_\psi$ must be a solution of $\S_\psi$. 
    So, either $x=\ZERO$ or $x \in \PosSol({\S_\psi})$. 
    Since $f^\psi$ is a circuit containing at least one not-gate, $\S_\psi$ contains a diff-constraint, and hence $x \ne \ZERO$. Thus $x \in \PosSol(\S_\psi)$, and since $\S_\psi$ tightly encodes $f^\psi$, we obtain $x \in \Fixe(f^\psi)$. 
    Second, suppose $x \in \Fixe(f^\psi)$. Then $x$ has an extension $x'\in\Sol(\S)$, obtained by setting $x'_v=\neg x_{\psi}$ and thus $x \in \PosSol(\S)_{I(\S_\psi)}$. 
\end{claimproof}

Let $G_{\psi,E}$ be the $\S$-extension of $H$; it is a loop-to-cycle digraph since $H$ is a loop-to-cycle digraph. Furthermore, since $\Delta^-(H)=2$, we have $\Delta^-(G_{\psi,E})=2$.

\begin{claim}
    If $\psi,E$ is a yes-instance of \qsat, then $G_{\psi,E}$ is UFP-dynamic.
\end{claim}
\begin{claimproof}
Suppose that $\psi,E$ is a positive instance of \qsat. Let $z\in\B^E$ such that $\psi$ is satisfied by all assignments extending $z$. Let $h\in F_\ZERO(H)$ defined as follows: 
\begin{itemize}
\item
$h_{\psi}(x)=0$,
\item
$h_\lambda(x)=x_{\psi}\lor x_v$ for all $\lambda \in E$ with $z_\lambda=1$,
\item
$h_\lambda(x)=x_{\psi}\land x_v$ for all $\lambda \in E$ with $z_\lambda=0$,
\item
$h_{a}(x)=x_a$ for every vertex $a$ distinct from $E\cup\{\psi\}$.
\end{itemize}
Suppose, for a contradiction, that there exists $x\in \Fixe^+(h)\cap \Sol^+(\S)$. Since $x\in \Sol^+(\S)$, $x_v \neq x_\psi$. Since $x\in \Fixe(h)$, we have $x_\psi=0$ and hence $x_v=1$. We deduce from the definition of $h$ that $x_E=z$. Since $x_\psi=0$ and since $\S$ encodes $f^\psi$, $x_\Lambda$ is not a satisfying assignment of $\psi$. However, since $x_E=z$, $x_\Lambda$ extends $z$, a contradiction. Thus $\Fixe^+(h)\cap \Sol^+(\S)=\emptyset$ and, by Theorem~\ref{thm:odd-extension}, $G_{\psi,E}$ is UFP-dynamic. 
\end{claimproof}

\begin{claim}
    If $G_{\psi,E}$ is UFP-dynamic, then $\psi, E$ is a yes-instance of \qsat.
\end{claim}
\begin{claimproof}
Conversely, suppose that $G_{\psi,E}$ is UFP-dynamic. By Theorem~\ref{thm:odd-extension}, there exists $h\in F_\ZERO(H)$ such that $\Fixe^+(h)\cap \Sol^+(\S)=\emptyset$. Let $z\in\B^E$ defined by $z_\lambda=h_\lambda(e_v)$ for all $\lambda\in E$. Let $z'\in\B^\Lambda$ extending $z$. Since $\S$ encodes $f^\psi$, there exists $x\in\Sol^+(\S)$ with $x_\Lambda=z'$ and $x_v\neq  x_\psi$. Let $x'\in\B^I$ such that $x'_\psi<x'_v$ and $x'_a=x_a$ for all $a\in I\setminus \{\psi,v\}$. Then $x'\in\Fixe^+(h)$ since $x'_\psi=0$ and, for all $\lambda\in E$, we have $h_\lambda(x')= h_\lambda(e_v) =z_\lambda=x_\lambda=x'_\lambda$. Since $x\neq x'$, we deduce that $x_\psi>x_v$, and, since $\S$ encodes $f^\psi$, $x_\Lambda=z'$ is a satisfying assignment of $\psi$. So every assignment extending $z$ is satisfying. Thus $\psi,E$ is a positive instance of \qsat.
\end{claimproof}
The main theorem follows.
\end{proof}

\newpagexorapx
\section{Complexity lower bound for \UFP}\label{sec:UFP}

Consider a 4-CNF formula $\psi$ over a set $\Lambda$ of variables and a set $M$ of clauses. By adding variables that do not appear in any clause, and by duplicating existing clauses, we can assume that the numbers of variables and clauses are powers of $2$, say $2^\ell$
and $2^m$ respectively (the resulting formula is of linear size in the size of the initial formula). Moreover, we can assume that all clauses have size exactly 4 by duplicating a literal within a clause as necessary.
Then, to simplify notations, we take two sets $V,C$ of size $\ell$ and $m$ respectively, and~write:
\[
\Lambda=\{\lambda_v\}_{v\in \B^V}, \quad M=\{\mu_c\}_{c\in \B^C},\qquad \mu_c=(\mu_{c,00},\mu_{c,01},\mu_{c,10},\mu_{c,11}).
\]
Thus, each clause is viewed as an ordered quadruple of literals (not necessarily distinct), rather than as a set of at most four literals. Also, we denote the positive and negative literals associated with $\lambda_v$ by $\EMM{\lambda_v^1}=\lambda_v$ and $\EMM{\lambda_v^0}=\neg \lambda_v$, respectively. Finally, it is convenient to regard an assignment of $\psi$ as a configuration $z$ on $\B^V$: a literal $\lambda_v^\sigma$ is \EM{true under} %
$z$ if $\sigma=z_v$. An assignment is then satisfying if, for each clause index $c\in\B^C$, the clause $\mu_c$ contains a literal which is true under $z$. 

A \EM{succinct representation} of $\psi$ is a circuit $f$ 
with the following specifications: the set of input vertices is $C\cup P$, where $P=\{p_0,p_1\}$, and the set of output vertices is $V\cup\{\rho\}$ (these four sets are pairwise disjoint); and for every fixed point $x$ of $f$ we have
\begin{equation}\label{eq:succinct}
    \mu_{x_C, x_P} = \lambda^{x_\rho}_{x_V}.
\end{equation}

Hence, each input configuration corresponds to a clause ($x_C$) and position in that clause ($x_P$), and the output configuration gives the corresponding literal: the involved variable ($x_V$) and the sign of the literal ($x_\rho$). 

It is possible to show by a lifting argument similar to that in \cite{PY86-projection} that our variant of {\succinct} is $\NEXPTIME$-complete, since $4$\textsc{SAT} (as above) inherits $\NP$-completeness from $3$\textsc{SAT} \cite{P94}.

Having defined the necessary technical concepts, we now introduce the principal problems of interest for this section.

\begin{quote}\decisionpb
	{\succinct}
	{a succinct representation of a 4-CNF formula $\psi$.}
	{is $\psi$ satisfiable?}
\end{quote}

\begin{quote}
\decisionpb
	{{\UFP} restricted to loop-to-cycle digraphs (\LCUFP)}
	{a loop-to-cycle digraph $G$.}
	{is $G$ UFP-dynamic?}
\end{quote}

We prove below that {\succinct} reduces to {\UFP}. Since {\succinct} is $\NEXPTIME$-complete \cite[Theorem 20.2]{P94}, this proves that {\UFP} is $\NEXPTIME$-hard. Actually, the reduction even shows that {\LCUFP} is $\NEXPTIME$-hard. Combined with Corollary~\ref{cor:in-np-or-nponp} this yields the first half of Theorem~\ref{thm:UFP}: $\NEXPTIME$-completeness of \UFP.

\begin{theoremrep}\label{thm:succinct-4cnf-iff-ufp}
Given a succinct representation of a 4-CNF formula $\psi$, we can construct in polynomial time a loop-to-cycle digraph $G_\psi$ such that:
\[
\text{$\psi$ is satisfiable} \iff \text{$G_\psi$ is UFP-dynamic} .
\] 
\end{theoremrep}

\begin{proofsketch}
The proof is similar to, though more involved than, that of Theorem~\ref{thm:qsat}, and again through \OSvsD. We construct a graph $H$, and odd-system $\S$, and $G = G_\psi$ which is the $\S$-extension of $H$. More precisely, Let $f^\psi$ be the succinct representation of $\psi$ and let $\S_\psi = \S_{ f^\psi }$ be the tight encoding of $f^\psi$. The vertices of $H$ are the variables of $\S_\psi$, which include the inputs $C \cup P$ and outputs $V \cup \{\rho\}$ of $f^\psi$, and three additional vertices  $s_0,s_1,\nu$. 
Figure \ref{fig:combined-ufp-complexity-construction} is an illustration of the graph $H$ and of the two constraints in $\S$ which are not also in $\S_\psi$. 
The proof shows that $\S$ does encode $\psi$, and that $\psi$ is satisfiable if and only if $G_\psi$ is UFP-dynamic.
\end{proofsketch}

\begin{toappendix}
    
\begin{figure}[!ht]
    \centering
    \includegraphics[page=3]{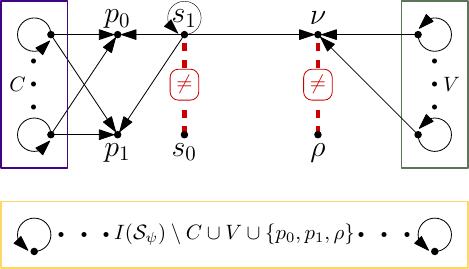}
    \caption{Graph $H$ from our construction. 
    Vertices in the yellow box at the bottom are in the graph $H$ but are relevant to the proof only through the constraints of $S_\psi$, which are not shown.}
    \label{fig:ufp-complexity-construction-v2}
\end{figure}
\end{toappendix}

\begin{proof}
Let $f^\psi$ be the succinct representation of $\psi$. Let $\S_\psi = \S_{ f^\psi }$ be the tight encoding of $f^\psi$. Let $H$ be the digraph defined as follows. The vertices are the variables of $\S_\psi$, which include the inputs $C \cup P$ and outputs $V \cup \{\rho\}$ of $f^\psi$, in addition to some other variables as necessary to guarantee the behavior prescribed above, and three additional vertices  $s_0,s_1,\nu$. The reader may find Figure \ref{fig:ufp-complexity-construction-v2} helpful in following the construction and proof. 
The arcs are:
\begin{itemize}
\item 
a loop on each vertex distinct from $s_0,p_0,p_1,\nu$,
\item 
an arc from each vertex in $C$ to each vertex in $\{p_0,p_1\}$,
\item
an arc from each vertex in $V$ to $\nu$,
\item
an arc from $s_1$ to each vertex in $\{p_0,p_1,\nu\}$. 
\end{itemize}

Let $\S$ be the odd-system defined as follows. The set of variables $I$ of $\S$ is the vertex set of $H$, and the constraint set consists of:
\begin{itemize}
\item 
all the constraints of $\S_\psi$, 
\item 
a diff-constraint between $s_0$ and $s_1$,
\item 
a diff-constraint between $\rho$ and $\nu$.
\end{itemize}

\begin{claim}\label{clm:ufp-s-encodes-psi}
    $\S$ encodes $f^\psi$, that is $\PosSol(\S)_{I(\S_\psi)}=\Fixe(f^\psi)$.
\end{claim}
\begin{claimproof}
    First, suppose $x \in \PosSol(\S)_{I(\S_\psi)}$. Clearly, $x \in \Sol(\S_\psi)_{I(\S_\psi)}=\Sol(\S_\psi)$ since the restriction of $x$ to the variable set of $\S_\psi$ must be a solution of $\S_\psi$. 
    So, either $x=\ZERO$ or $x \in \PosSol({\S_\psi})$. 
    Since $f^\psi$ is a circuit containing at least one not-gate, $\S_\psi$ contains a diff-constraint, and hence $x \ne \ZERO$. Thus $x \in \PosSol(\S_\psi)$, and since $\S_\psi$ tightly encodes $f^\psi$, we obtain $x \in \Fixe(f^\psi)$. 
    Second, suppose $x \in \Fixe(f^\psi)$. Then $x$ has exactly two extensions $x^0,x^1\in\Sol(\S)$, obtained by setting $x^0_{\nu}=x^1_{\nu}=\neg x_{\rho}$, $x^0_{s_0}=x^1_{s_1}=0$ and $x^0_{s_1}=x^1_{s_0}=1$, and thus $x \in \PosSol(\S)_{I(\S_\psi)}$. 
    The claim follows.
\end{claimproof}

Let $G_\psi$ be the $\S$-extension of $H$; it is a loop-to-cycle digraph since $H$ is a loop-to-cycle digraph. 

\begin{claim}\label{clm:psi-sat-implies-ufp}
    If $\psi$ has a satisfying assignment then $G_\psi$ is UFP-dynamic.
\end{claim}
\begin{claimproof}
Suppose that $\psi$ has a satisfying assignment $z$. 
Then there exists words $\alpha, \beta$ such that for all $c \in \B^C$, the literal $\mu_{c,\alpha_c\beta_c}$ is true under $z$.
We define $h\in F_\ZERO(H)$ as follows: 
\begin{itemize}
\item
$h_{s_0}(x)=0$,
\item
$h_{p_0}(x)=\alpha_{x_C}$ if $x_{s_1}=1$ and $h_{p_0}(x)=\sum_{c \in C} x_c$  otherwise, 
\item
$h_{p_1}(x)=\beta_{x_C}$ if $x_{s_1}=1$ and $h_{p_1}(x)=\sum_{c \in C} x_c$ otherwise, 
\item
$h_{\nu}(x)=z_{x_V}$ if $x_{s_1}=1$ and $h_{\nu}(x)=\sum_{v \in V} x_v$ otherwise, 
\item
$h_{a}(x)=x_a$ for every vertex $a$ distinct from $s_0,p_1,p_2,\nu$. 
\end{itemize}

For the sake of reader intuition, we note that the definition of $h_a$ is forced by $h \in F_\ZERO(H)$, and also that the intention behind setting $h_i = \sum_{a \in A} x_a$ is simply to ensure that $h\in F(H)$, i.e. to realize all necessary essential dependencies.

Suppose for contradiction that $x\in \Fixe^+(h)\cap \PosSol(\S)$.
Since $x\in \Sol^+(\S)$ we have $x_{s_0}\neq x_{s_1}$, $x_{\rho}\neq x_\nu$ and $\mu_{x_C,x_P}=\lambda^{x_\rho}_{x_V}$ since $\S$ encodes $f^{\psi}$. Since $x\in \Fixe(h)$ we have $x_{s_0}=0$ and hence $x_{s_1}=1$.
We deduce from the definition of $h$ that $x_P=\alpha_{x_C}\beta_{x_C}$ and $x_{\nu}=z_{x_V}$. \\
Thus the variable $\mu_{x_C,x_P}=\mu_{x_C,\alpha_{x_C}\beta_{x_C}}$ is true under $z$ by definition of $\alpha$ and $\beta$. 
We also have that $\mu_{x_C,x_P}=\lambda^{x_\rho}_{x_V}$ since $x\in \Sol(S)$. Thus $z_{x_V}=x_\rho$, and consequently $x_\rho=x_\nu$, yielding the desired contradiction. Thus $\Fixe^+(h) \cap \PosSol(\S) = \emptyset$ and, by Theorem~\ref{thm:odd-extension} $G_\psi$ is UFP-dynamic.
\end{claimproof}

\begin{claim}\label{clm:ufp-implies-psi-sat}
    If $G_\psi$ is UFP-dynamic then $\psi$ has a satisfying assignment.
\end{claim}
\begin{claimproof}
If $G_\psi$ is UFP-dynamic, then by Theorem~\ref{thm:odd-extension}, there exists $h\in F_\ZERO(H)$ such that $\Fixe^+(h)\cap \Sol^+(\S)=\emptyset$.

For every clause index $c\in\B^C$, $h$ has a fixed point $x$ extending $c$ with $x_{s_0}<x_{s_1}$. We then refer to $x_P$ as %
the position computed by $h$ for the clause $\mu_c$, and denote it $p(c)$; this is well-defined since the in-neighborhood  of $p_0$ and $p_1$ is $C\cup\{s_1\}$. 
For every variable index $v\in\B^V$, $h$ has a fixed point $x$ extending $v$ with $x_{s_0}<x_{s_1}$. 
We then refer to $x_{\nu}$ as the assignment of $\lambda_v$ computed by $h$; this is again well-defined since the in-neighborhood of $\nu$ and $p_1$ is $V\cup\{s_1\}$. Let $z$ the assignment of $\psi$ such that, for all $v\in\B^{V}$, $z_v$ is the assignment  of $\lambda_v$ computed by $h$. We will prove that $z$ is a satisfying assignment.

Let $c\in\B^{C}$ and let $x\in\Sol^+(\S)$ with $x_{s_0}<x_{s_1}$, $x_C=c$ and $x_P=p(c)$. Because $x\in\Sol^+(\S)$ and $\Fixe^+(h) \cap \PosSol(\S) = \emptyset$, we know that $x \notin \Fixe(h)$.
Then $x_{\nu}\neq x_{\rho}$ and since $\S$ encodes $f^\psi$, we have $\mu_{c,p(c)}=\lambda_v^{x_{\rho}}$ where $v=x_{V}$. 
Since $x\not\in\Fixe(h)$, we deduce that $h_\nu(x) \ne x_\nu$ and so $x + e_\nu \in \Fixe^+(h)$,
thus $x_{\nu}+1$ is the assignment of $\lambda_v$ computed by $h$, that is, $z_v=x_{\nu}+1$. Since $x_{\nu}\neq x_{\rho}$, we deduce that $z_v=x_{\rho}$. Thus clause $c$ contains a literal which is true under $z$, namely $\mu_{c,p(c)}=\lambda_v^{x_{\rho}}$.
We deduce that $z$ is satisfies $\psi$.  
\end{claimproof}

Combining Claims \ref{clm:psi-sat-implies-ufp} and \ref{clm:ufp-implies-psi-sat}, the main result follows.
\end{proof}

\newpagexorapx
\section{Complexity lower bounds for \PeriodPiece{} and \texorpdfstring{$d$-\PeriodPiece}{d-Period Piece}}

Let $\Lambda$ be a non-empty finite subset of $\N^*$. 
For any transformation $f: X \to X$, let $\Per( f )$ denote the set of periodic points of $f$, and $\Per_\Lambda( f )$ the set of periodic points of $f$ on a cycle whose length belongs to $\Lambda$. 
For instance, $\Per_{ \{ 1 \} }( f ) = \Fixe( f )$. 
We shall abuse notation slightly, and write $\Per_\Lambda( f ) = \sigma$ (or $\Per(f)=\sigma$) to mean that $f$ acts as a permutation of $\Per_\Lambda( f )$ (or $\Per(f)$) that is conjugate to $\sigma$. 

We consider the set of permutations $\sigma$ of some non-empty finite set $X = \{1, \dots, k\}$ such that all the cycles of $\sigma$ have lengths in $\Lambda$, which we denote by $S_\Lambda$.
For instance, if $\Lambda = \{ 1 \}$, then $S_\Lambda = \{ \id_k : k \in \N^* \}$, where $\id_k$ is the identity on $k$ elements. 
Let $\Sigma$ be a non-empty finite subset of $S_\Lambda$. 
We say that $f$ is \EM{Period Piece$(\Lambda, \Sigma)$} if $\Per_\Lambda( f ) = \sigma$ for some $\sigma \in \Sigma$.

\begin{quote}
\decisionpb
	{($d$-)$\PeriodPiece(\Lambda, \Sigma)$}
	{a digraph $G$ (with $\Delta^-(G)\leq d$).}
	{is $G$ Period Piece$(\Lambda, \Sigma)$-dynamic?}
\end{quote}

This formalism can express many dynamical properties about periodic points. 
For instance, if $\Lambda = \{ 1 \}$ and $\Sigma = \{ \id_1, \dots, \id_k \}$, Period Piece$(\Lambda, \Sigma)$ is equivalent to ``Between $1$ and $k$ fixed points.'' 
Moreover, if $\Lambda = \{ \ell \}$ and $\Sigma = \{ C_\ell \}$ a cyclic permutation of degree $\ell$, it is equivalent to ``Unique cycle of length $\ell$.''
Hence, if $\Lambda = \{ 1 \}$ and $\Sigma = \{ \id_1 \}$, then $\PeriodPiece(\Lambda, \Sigma)$ and $d$-$\PeriodPiece(\Lambda, \Sigma)$ are $\UFP$ and $\dUFP$, respectively. 
The following lemma allows us to apply the complexity upper bounds from Theorem~\ref{thm:in-nexptime} and Corollary~\ref{cor:in-np-or-nponp}:

\begin{lemmarep} \label{proposition:lambda_in_sigma_Sigma2}
For any $\Lambda$ and $\Sigma$ as above, the Period Piece$(\Lambda, \Sigma)$-property belongs to $\mathrm{FO}\Sigma_2$.
\end{lemmarep}

\begin{proof}
The proof goes bottom up. Firstly, the predicate $[ x \in \Per_{\{ \lambda \}} ]$ can be expressed in $\mathrm{FO}\Pi_1$ as $[ x \in \Per_{\{ 1 \} }] \equiv x \to x$ for $\lambda = 1$ and
\[
    [x \in \Per_{\{ \lambda \}}] \equiv \forall y_1, \dots, y_\lambda \ \left[ (x = y_1 ) \land ( y_1 \to y_2 \to \dots \to y_\lambda ) \right] \implies \left[ ( y_i \ne y_j \ \forall i \ne j ) \land ( y_\lambda \to x ) \right]
\]
for all $\lambda \ge 2$. Thus, the predicate $[x \in \Per_\Lambda]$ can be expressed in $\mathrm{FO}\Pi_1$ as the disjunction 
\[
    [ x \in \Per_{ \Lambda } ] \equiv \bigvee_{ \lambda \in \Lambda } [ x \in \Per_{ \{ \lambda \} } ].
\]
Secondly, let $\sigma$ be a finite permutation of $\{1, \dots, k\}$. Then the property $[\Per_\Lambda = \sigma]$ can be expressed in $\mathrm{FO}\Sigma_2$ as
\[
    [ \Per_\Lambda = \sigma ] \equiv \exists y_1, \dots, y_k \forall x \ (y_i \ne y_j \ \forall i \ne j) \land (y_i \to y_{ \sigma( i ) } \ \forall 1 \le i \le k) \land \left[ [ x \in \Per_\Lambda ] \implies \bigvee_{i=1}^k ( x = y_i ) \right].
\]
Finally, $[ \Per_\Lambda \in \Sigma ]$ can be expressed in $\mathrm{FO}\Sigma_2$ as the disjunction 
\[
    [ \Per_\Lambda \in \Sigma ] \equiv \bigvee_{ \sigma \in \Sigma } [ \Per_\Lambda = \sigma ].
\]
\end{proof}

The following theorem, taken together with Theorem~\ref{thm:in-nexptime} and Corollary \ref{cor:in-np-or-nponp}, yields Theorem~\ref{theorem:perlambda}: $\NEXPTIME$-completeness of $\PeriodPiece(\Lambda, \Sigma)$ and $\NPoNP$-completeness of $d$-$\PeriodPiece(\Lambda, \Sigma)$.

\begin{theoremrep}\label{thm:perlambda-in-sigma-completeness}
For every non-empty finite subset $\Lambda$ of $\N^*$ and every non-empty finite subset $\Sigma$ of $S_\Lambda$, $\PeriodPiece(\Lambda, \Sigma)$ is $\NEXPTIME$-hard and $d$-$\PeriodPiece(\Lambda, \Sigma)$ is $\NPoNP$-hard for all $d\geq 2$. 
\end{theoremrep}

We note that for the complexity lower bound, we do require $\Sigma$ to be non-empty and finite. Let $\Lambda = \{ 1 \}$ for instance. If on the one hand $\Sigma$ is empty, then the Period Piece$(\Lambda, \Sigma)$-property is vacuous. If on the other hand $\Sigma = \{ \id_k : k \ge 1 \}$, then the Period Piece$(\Lambda, \Sigma)$-property is ``At least one fixed point,'' and every digraph is a yes-instance.
Theorem~\ref{thm:perlambda-in-sigma-completeness} follows from Theorem~\ref{thm:perlambda_in_sigma} below.

\begin{theoremrep}\label{thm:perlambda_in_sigma}
Let $\Lambda$ be a non-empty finite subset of $\N^*$ and $\Sigma$ be a non-empty finite subset of $S_\Lambda$. Let $G$ be a loop-to-cycle digraph with at least one terminal cycle and $\Delta^-(G) = 2$. We can construct in polynomial time a digraph $G_{\Lambda, \Sigma}$, with $\Delta^-( G_{\Lambda,\Sigma} ) \leq 2$, such that
\[
    \text{$G$ is UFP-dynamic} \iff \text{$G_{\Lambda,\Sigma}$ is Period Piece$(\Lambda, \Sigma)$-dynamic}.
\] 
\end{theoremrep}

\begin{proofsketch}
The graph $G_{\Lambda,\Sigma}$ is constructed in three steps.
\begin{itemize}
    \item First, repeatedly subdivide the arcs in the terminal cycles of $G$ to obtain $G'$. Then $G$ is UFP-dynamic if and only if $G'$ is UFP-dynamic. Also, for any $f \in F(G')$, all the periodic configurations of $f$ are either fixed points or have very large period. This ensures that if $G$ is UFP-dynamic, then any $f$ over $G_{\Lambda,\Sigma}$ will not have any unwanted short cycles in $\Per_\Lambda( f )$.

    \item Second, make many disjoint copies $G'$, say $G'' = G'_1 \cup \dots \cup G'_n$. This ensures that if $G$ is not UFP-dynamic, then any $f$ over $G_{\Lambda,\Sigma}$ will have either too many cycles in $\Per_\Lambda( f )$ or none at all for a given cycle length.

    \item Third, let $\sigma \in \Sigma$, then there exists a Boolean network $h$ whose interaction graph $H_\sigma$ has in-degree $2$ and $\Per( h ) = \sigma$. Thus, $G_{\Lambda, \Sigma}$ is the disjoint union $G'' \cup H_\sigma$. This ensures that if $G$ is UFP-dynamic, then choosing $f^i$ with a unique fixed point over every copy $G'_i$ and $h$ over $H_\sigma$ will result in $f$ over $G_{\Lambda,\Sigma}$ with exactly $\Per_\Lambda( f ) = \sigma$.
\end{itemize}
The result follows. \end{proofsketch}

\begin{toappendix}
We need several lemmas.

\begin{lemmarep}\label{lem:boolean_sigma}
Let $\sigma$ be a permutation of degree $k \ge 1$. There exists a graph $H_\sigma$ of maximum in-degree $2$ which is from $\sigma$ and a Boolean network $f\in F(H_\sigma)$
such that $\Per( f ) = \sigma$.
\end{lemmarep}

\begin{proof}
The basic idea of the construction is as follows. In the cyclic shift network of size $n \ge 1$, i.e. $h_i(x) = x_{i-1}$ with indices computed cyclically, the unit configurations $e_i$ form a cycle of period $n$. Since $\sigma$ can be decomposed into cycles, say it has $n_\ell$ cycles of length $\ell$, then we use $n_\ell$ cyclic shift networks of size $\ell$ to obtain $\sigma$ thanks to the unit configurations. Any non-unit configuration will be mapped to $e_1$, thanks to a circuit that feeds back to the cycles.

More formally, let $\sigma$ have $n_\ell \ge 0$ cycles of length $\ell$ for $1 \le \ell \le k$, where $\sum_{\ell = 1}^k \ell\cdot n_\ell = k$.
Let $g$ be a layered circuit \cite{GJ11} that computes the Boolean function $\psi : \B^k \to \B$ given by
\[
    \psi( x_1, \dots, x_k ) = \begin{cases}
        1 & \text{if } x = e_i \text{ for some } 1 \le i \le k \\
        0 & \text{otherwise}.
    \end{cases}
\]
As usual, the output vertex of $g$ is denoted $\psi$. Say $d$ is the depth of the circuit $g$. Let $H_\sigma=H$ be the digraph obtained from $g$ by
\begin{itemize}
    \item removing the loops on the inputs $I = \{1, \dots, k\}$ of $g$,

    \item adding arcs so $H[I]$ is exactly a disjoint union of $n_\ell$ cycles for all $1 \le \ell \le k$,

    \item adding an arc from $\psi$ to each vertex in $I$.
\end{itemize}

\begin{figure}
    \centering
    \includegraphics[page=6]{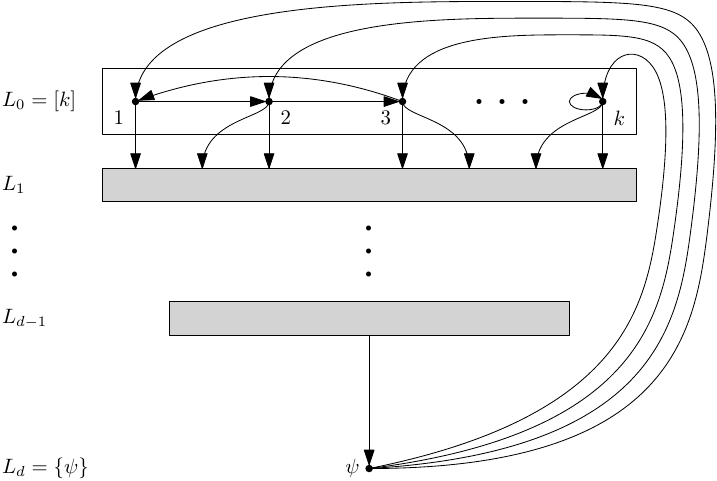}
    \caption{The graph $H_\sigma$ from our construction for Lemma \ref{lem:boolean_sigma}.}
    \label{fig:lem-boolean-sigma}
\end{figure}

The construction is illustrated in Figure \ref{fig:lem-boolean-sigma}.
We note that $\Delta^-( H ) = 2$. Let $h \in F( H[I] )$ be given by $h_i( x ) = x_{i'}$, where $i'$ is the predecessor of $i$ on its cycle. Note that $h$ is a permutation of $\{0,1\}^I$. Let $f \in F( H )$ be defined by
\begin{itemize}
    \item $f_1( x ) = \lnot x_\psi \lor h_1( x_I )$,
    \item $f_i( x ) = x_\psi \land h_i( x_I )$ for all $i \in I \setminus \{ 1 \}$,
    \item $f_a( x ) = g_a( x )$ for all $a \notin I$.
\end{itemize}

We now prove that the periodic part of $f$ is conjugate to $\sigma$. Let $x$ be a configuration on $H$. If $x_I$ is a unit configuration, then so is $f( x )_I$. If $x_I$ is not a unit configuration, then $f^d( x )_\psi = 0$, and hence $f^{d+1}( x )_I = e_1$. In any case, we obtain that $f^{d+t}( x )_I$ is a unit configuration for all $t \ge 1$.  Thus, for any periodic configuration $x$, $x_I$ is a unit configuration.

Next, if $z$ is a unit configuration on $I$, then there is a unique periodic configuration $\tilde{z}$ of $f$ that extends $z$. Explicitly, let the layers of $g$ be $L_0, L_1, \dots, L_d$, where $L_0 = I$ and $L_d = \{ \psi \}$. Since the circuit is acyclic, each layer $L_\delta$ can be viewed as computing a cumulative function of the inputs, say $\hat{g}_{ L_\delta } = g_{ L_\delta } \circ g_{ L_{\delta - 1} } \circ \cdots \circ g_{ L_1 } : \B^I \to \B^{ L_\delta }$. With this notation, we obtain $\tilde{z}_{L_\delta} = \hat{g}_{ L_\delta }( h^{-\delta}( z ) )$. If $i \in I$ belongs to a cycle of length $\ell$ in $H$, then $\tilde{e_i}$ is a periodic point of $f$ of period $\ell$ accordingly.

We obtain that the set of periodic points of $f$ is $\{ \tilde{e_i} : i \in I \}$ and that the periodic part of $f$ is conjugate to $\sigma$.
\end{proof}

Let $G$ be a digraph. \EM{Subdividing} an arc of $G$ from $a$ to $b$ consists in removing this arc, adding a new vertex $c$ together with an arc from $a$ to $c$ and from $c$ to $b$. A \EM{subdivision} of $G$ is a digraph obtained from $G$ by applying successive subdivisions.

\begin{lemmarep}\label{lem:subdivision}
A digraph is UFP-dynamic if and only if all its subdivisions are UFP-dynamic. 
\end{lemmarep}

\begin{proof}
Let $H$ be a digraph, and let $G$ be obtained from $H$ by subdividing an arc from $a$ to $b$; we denote by $c$ the additional vertex. It is sufficient to prove that $H$ is UFP-dynamic if and only if $G$ is UFP-dynamic. Actually, by Lemma~\ref{lem:UFP0}, it is sufficient to prove that $F_\ZERO(H)$ has a network with unique fixed point if and only if $F_\ZERO(G)$ has a network with a unique fixed point. 

Let $h\in F(H)$ with $\ZERO_H$ as unique fixed point. Given $x\in\B^{G}$, we set $\tilde x=x$ if $x_a=x_c$ and $\tilde x=x+e_a$ otherwise, so that $\tilde x_a=x_c$. Let $f\in F(G)$ defined by $f_v(x)=h_v(x_H)$ for all $v\neq b,c$, $f_b(x)=h_b(\tilde x_H)$ and $f_c(x)=x_a$. Since $\tilde \ZERO_G=\ZERO_G$, we have  $f(\ZERO_G)_H=h((\ZERO_G)_H)=h(\ZERO_H)=\ZERO_H$ and $f_c(\ZERO_G)=0$, thus $\ZERO_G$ is a fixed point of $f$. Suppose that $f(x)=x$. Since $x_c=f_c(x)=x_a$, we have $\tilde x=x$. Hence $x_H=f(x)_H=h(x_H)$, which implies $x_H=\ZERO_H$. We deduce that $x_c=x_a=0$ and thus $x=\ZERO_G$. So $\ZERO_G$ is the unique fixed point of $f$. 

Conversely, let $f\in F(G)$ with $\ZERO_G$ as unique fixed point. Since $a$ is the unique in-neighbor of $c$ in $G$, and $f_c(\ZERO_G)=0$, we deduce that $f_c(x)=x_a$ for all $x\in\B^G$. Given $x\in\B^H$, we denote by $\tilde x$ the extension of $x$ on $G$ such that $\tilde x_c=x_a$. Let $h\in F(H)$ defined by $h(x)=f(\tilde x)_H$. We have $h(\ZERO_H)=f(\ZERO_G)_H=(\ZERO_G)_H=\ZERO_H$. Suppose that $h(x)=x$. Then $\tilde x_H=x=h(x)=f(\tilde x)_H$. Furthermore, $f_c(\tilde x)=x_a=\tilde x_c$. Thus $f(\tilde x)=\tilde x$, which implies $\tilde x=\ZERO_G$ and thus $x=\ZERO_H$. So $\ZERO_H$ is the unique fixed point of $h$. 
\end{proof}

\begin{lemmarep}\label{lem:power}
Let $G$ be a digraph such that no two cycles intersect in a vertex and every cycle has length a power of $2$. Then every limit cycle of every $f\in F(G)$ has length a power of $2$.
\end{lemmarep}

\begin{proof}
If $G$ is acyclic, then $f$ has a unique limit cycle, which is of length one (Lemma~\ref{lem:robert}), so we are done. Let $G$ have a cycle. Let $f\in F(G)$ and let $f^2=f\circ f$.  

\begin{claim}\label{clm:power-1}
If every limit cycle of $f^2$ has length a power of $2$, then every limit cycle of $f$ has length a power of $2$.
\end{claim}

\begin{claimproof}
Suppose that $f$ has a limit cycle of length $q$, and let $X$ be its set of configurations. If $q$ is odd then $f^2$ has a limit cycle of length $q$ whose set of configurations is also $X$, so the claim vacuously holds. So $q=2p$ for some $p\geq 1$. Then there is a partition of $X$ in two sets $X_0,X_1$ of equal size such that $f^2$ has two limit cycles of length $p$ whose set of configurations are $X_0$ and $X_1$. If $p$ is a power of two, then $q$ is a power of two and the claim holds. Otherwise, $p$ is not a power of two and the claim vacuously holds.
\end{claimproof}

We denote by $\ell(G)$ the maximum length of a cycle in $G$. We denote by $N(G)$ the set of vertices $i$ in $G$ that have an out-neighbor $j\neq i$. Let $G^2$ be the digraph with vertex set $V(G)$ and an arc from $a$ to $b$ if there exists $c$ (not necessarily distinct from $a,b$) such that $G$ has an arc from $a$ to $c$ and from $c$ to $b$. Let $G^{(2)}$ be the interaction graph of $f^2$. 

\begin{claim}\label{clm:power-2}
Let $G$ be as in the lemma statement. Then the following holds:
\begin{itemize}
\item \label{item:G2_satisfies}
$G^2$ satisfies the conditions of the lemma, %
\item \label{item:l(G2)}
$\ell(G^2)=\max(1,\ell(G)/2)$,
\item \label{item:N(G2)}
$N(G^2)\subseteq N(G)$,
\item \label{item:G2_subgraph}
$G^{(2)}\subseteq G^2$.
\end{itemize}
\end{claim}
\begin{claimproof}
Let the set of loops of $G$ be $L$. Denote the other cycles of $G$ as $C^s = \{ v^s_0, v^s_1, \dots, v^s_{2 n_s - 1} \}$ for all $s \in S$. Then the collection of cycles of $G^2$ is $L \cup \{ C_0^s, C_1^s : s \in S \}$, where $C_0^s = \{ v^s_0, v^s_2, \dots, v^s_{2 n_s - 2} \}$ and $C_1^s = \{ v^s_1, v^s_3, \dots, v^s_{2 n_s - 1} \}$ for all $s \in S$. The first two items immediately follow.

We now prove the third item. If $i \in N(G^2)$, then there exist $k$ and $j \ne i$ such that $(i,k)$ and $(k,j)$ are arcs of $G$. If $k \ne i$, then $i \in N(G)$ because of $(i,k)$; if $k = i$, then $i \in N(G)$ because of $(k,j)$.

We finally prove the fourth item. For any configuration $c \in \B^{V(G)}$, let $c^1 = f( c )$ and $c^2 = f^2( c )$ and denote $c_{-i} = ( c_v : v \in V(G) \setminus\{ i \} )$. 
Let $(i,j)$ be an arc of $G^{(2)}$.
Then there exist $a, b \in \B^{V(G)}$ such that $a_{-i} = b_{-i}$ and $a^2_j \ne a^2_j$. Since $a^2_j \ne b^2_j$, there exists $k$ such that $(k,j) \in E( G )$ and $a^1_k \ne b^1_k$. But then, $a_{-i} = b_{-i}$ and $a^1_k \ne b^1_k$ imply that $(i,k) \in E(G)$ as well. Thus $(i,j)$ is an arc of $G^2$.
\end{claimproof}

We now prove the lemma by induction on $\ell(G)$, and then on $N(G)$.  Suppose that $\ell(G)=1$. If $N(G)=\emptyset$, then all the arcs of $G$ are loops, and we easily check that every limit cycle of $f$ has length $1$ or $2$. Suppose that $N(G)\neq\emptyset$. Since $\ell(G)=1$, there exists a vertex $i\in N(G)$ such that either $i$ is a source or has itself as unique in-neighbor. 

Suppose that $i$ is a source. Then there exists $\alpha$ such that $f_i(x)=\alpha$ for all $x\in\B^G$. Let $x'=x$ if $x_i=\alpha$, and $x'=x+e_i$ otherwise, and let $f' : \B^G\to\B^G$ defined by $f'(x)=f(x')$. Then the interaction graph $G'$ of $f'$ is a subgraph of $G$ and $i$ has no out-neighbor in $G'$. Thus $N(G')\subsetneq N(G)$ and, by induction, every limit cycle of $f'$ has length a power of $2$. Since $f$ and $f'$ have the same limit cycles we are done.

If $i$ is not a source (i.e., $i$ has itself as unique in-neighbor), then $f^2_i(x)=x_i$ for all $x\in\B^G$. By Claim \ref{clm:power-2}, $\ell(G^{(2)})\leq 1$ and $N(G^{(2)})\subseteq N(G^2)\subseteq N(G)$. For all $x\in \B^G$ and $\alpha\in\B$, let $x^\alpha=x$ if $x_i=\alpha$, and $x^\alpha=x+e_i$ otherwise. Let $f^{(\alpha)}:\B^G\to\B^G$ defined by $f^{(\alpha)}(x)=f(x^\alpha)$.  Then the interaction graph $G^\alpha$ of $f^{(\alpha)}$ is a subgraph of $G^2$ and $i$ has no out-neighbor in $G^\alpha$. Thus $N(G^\alpha)\subsetneq N(G)$ and, by induction, every limit cycle of $f^{(\alpha)}$ has length a power of $2$. Since every limit cycle of $f^2$ is a limit cycle of $f^{(0)}$ or $f^{(1)}$, we deduce that every limit cycle of $f^2$ has length a power of $2$. Thus, by Claim \ref{clm:power-1}, every limit cycle of $f$ has length a power of $2$. This completes the case $\ell(G)=1$. 

Suppose that $\ell(G)\geq 2$. By Claim \ref{clm:power-2} and the induction hypothesis, every limit cycle of $f^2$ has length a power of $2$. Thus, by Claim \ref{clm:power-1}, every limit cycle of $f$ has length a power of $2$.
\end{proof}

\begin{lemmarep}\label{lem:long_cycle}
Let $G$ be a loop-to-cycle digraph with at least one terminal cycle. Suppose that all the terminal cycles have length $2^\ell$. Let $f\in F(G)$ with a unique fixed point. All the periodic configurations of $f$ distinct from the fixed point have period at least $2^{\ell+1}$. 
\end{lemmarep}

\begin{proof}
Let $I$ be the vertices of $G$ that are not in the terminal cycles, and let $L$ be the vertices with a loop. Let $x$ be the unique fixed point of $f$, and let $z$ be a periodic configuration of $f$ distinct from $x$. Then $x_I,z_I$ are fixed global contexts. If $z_I=x_I$ then, by Lemma~\ref{lem:global_context}, $z$ converges to $x$, a contradiction. Thus $z_I\neq x_I$. If $z_I$ is not a negative context for at least one terminal cycle, then, by the first and second points of  Lemma~\ref{lem:global_context}, $f$ has at least two fixed points, a contradiction. Hence $z_I$ is a negative context for at least one terminal cycle, of length $2^\ell$ by hypothesis. The unique divisor of $2^{\ell+1}$ which does not divide $2^{\ell}$ is $2^{\ell+1}$. Hence, by the third point of Lemma~\ref{lem:global_context}, $z_I$ converges to a periodic configuration whose period is a multiple of $2^{\ell+1}$. 
\end{proof}

Given $n$ transformations $f^i$, with domain $X^i$, their \EM{product}, denoted \EM{$(f^1, \dots, f^n)$}, is the transformation with domain $X=X^1\times \cdots\times X^n$ defined by 
\[
f(x^1,\dots,x^n) = (f^1(x^1),\dots,f^n(x^n)). 
\]
In graphical terms, the graph $D_f$ is then exactly the Cartesian product of the graphs $D_{f^1},\ldots, D_{f^n}$. 
To assist with intuition, the Cartesian product of two graphs is shown in \ref{fig:cartesian-product}.
\begin{figure}[!ht]
    \centering
    \includegraphics[page=7,width=.5\linewidth]{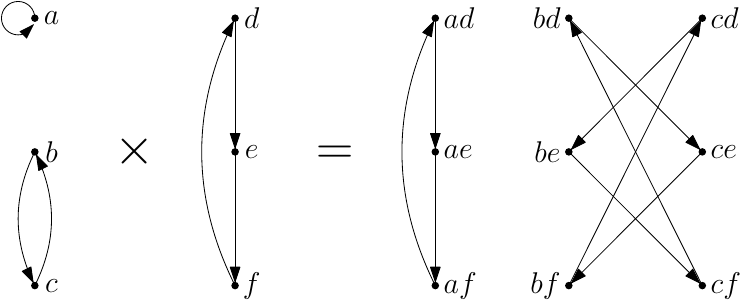}
    \caption{Graphical representation of the product of a transformation of $\{a,b,c\}$ and a transformation of $\{d,e,f\}$.}
    \label{fig:cartesian-product}
\end{figure}

If $\sigma_1,\sigma_2$ are cyclic permutations on $\ell_1$ and $\ell_2$ elements, respectively, then $(\sigma_1,\sigma_2)$ is a permutation that consists of $\ell_1\ell_2/\lcm(\ell_1,\ell_2)$ limit cycles of length $\lcm(\ell_1,\ell_2)$ \cite{DFMR23}. From this observation, we deduce that, for all $\ell\geq 1$,  
\begin{equation}\label{eq:product}
    n_\ell ( f^1, \dots, f^n )=\frac{1}{\ell}
\sum_{\substack{\ell_1,\dots,\ell_n\in\N\\\lcm(\ell_1,\dots,\ell_n)=\ell}}
\prod_{i=1}^n \ell_i\cdot n_{\ell_i}(f^i).
\end{equation}

\begin{lemmarep}\label{lem:two-copies}
Let $\ell \ge 1$ and let $f^1,\dots,f^{n+1}$ be transformations. Suppose that, for $1\leq i\leq n$, $f^i$ does not have a unique fixed point and every limit cycle of $f^i$ has length a power of $2$. Let $f=(f^1, \cdots, f^{n+1})$. If $f$ has a limit cycle of length $\ell$, then $f$ has at least $2^{n-1}$ limit cycles of length~$\ell$. 
\end{lemmarep} 

\begin{proof}
Let $k$ such that $\ell=2^k\cdot p$ for some odd integer $p$. By \eqref{eq:product}, there exists $\ell_1,\dots,\ell_{n+1}\in\N$ with $\lcm(\ell_1,\dots,\ell_{n+1})=\ell$ and $n_{\ell_i}(f^i)\geq 1$ for $1\leq i\leq n+1$. By hypothesis, for all $1\leq i\leq n$, we have $\ell_i=2^{k_i}$ for some $k_i\in\N$ and thus $\ell_{n+1}=2^{k_{n+1}}\cdot p$ for some $k_{n+1}\in\N$; furthermore, $\max_{i=1}^{n+1} k_i=k$. If $k_{n+1}<k$ we can suppose, without loss, that $k_n=k$. Thus we can suppose that $\ell_n\cdot\ell_{n+1}\geq 2^k\cdot p=\ell$. Consequently, by \eqref{eq:product} we have
\[
n_\ell(f)
\geq \frac{1}{\ell}\prod_{i=1}^{n+1} \ell_i\cdot n_{\ell_i}(f^i)
\geq \prod_{i=1}^{n-1} \ell_i\cdot n_{\ell_i}(f^i).
\]
Let $1\leq i<n$. If $\ell_i=1$ then $n_{\ell_i}(f^i)\geq 2$ since $f^i$ does not have a unique fixed point. Thus $\ell_i \cdot n_{\ell_i}(f^i)\geq 2$ in every case. We deduce that $n_\ell(f)\geq 2^{n-1}$. 
\end{proof}

\begin{lemmarep}\label{lem:unique-copy}
Let $f,h$ be transformations and $\ell\in\N$ with $n_\ell(h)\geq 1$. Suppose that $f$ has a unique fixed point and that there is no integer $2\leq p\leq\ell$ with $n_p(f)\geq 1$. Then $n_\ell(f,h)=n_\ell(h)$. 
\end{lemmarep}

\begin{proof}
Let $\Omega$ be the set of couples $(p,q)\in\N^2$ with $\lcm(p,q)=\ell$ and $n_p(f)n_q(h)\geq 1$; since $n_1(f)=1$ and $n_\ell(h)\geq 1$ we have $(1,\ell)\in\Omega$. If $(p,q)\in\Omega$, then $p\leq \ell$, and thus $p=1$, which forces $q=\ell$. Thus $\Omega=\{(1,\ell)\}$ and we deduce from \eqref{eq:product} that $n_\ell(f,h)=\frac{1}{\ell}n_1(f)\ell n_\ell(h)=n_\ell(h)$. 
\end{proof}
\end{toappendix}

\begin{toappendix}
\begin{proof}[\bf Proof of Theorem~\ref{thm:perlambda_in_sigma}]
Recall our instance $G$ of $\UFP$ is a loop-to-cycle digraph.
The construction of $G_{\Lambda, \Sigma}$ is as follows. 
Let $k$ be an integer greater than any element of $\Lambda$ and greater than the degree of any permutation in $\Sigma$. 
Let $m$ be the length of a longest terminal cycle in $G$ and let $N = \max(\lceil\log_2 m\rceil,\lceil\log_2 k\rceil)$. 
Let $G'$ be obtained from $G$ by applying successive subdivisions of arcs in the terminal cycles until each terminal cycle has length $2^N$. 
By Lemma~\ref{lem:subdivision}, $G'$ is UFP-dynamic if and only if $G$ is UFP-dynamic. Let $n = \lceil\log_2 k\rceil + 2$, and let $G'_1, \dots, G'_n$ be disjoint copies of $G'$. 
Let $\sigma \in \Sigma$ and $H_\sigma$ be as in Lemma~\ref{lem:boolean_sigma}. 
Thus $H_\sigma$ is Period Piece$(\Lambda, \Sigma)$-dynamic and maximum in-degree $2$. 
Finally, $G_{\Lambda, \Sigma}$ is the disjoint union of $G'_1, \dots, G'_n, H_\sigma$. 

We now prove that $G'$ is UFP-dynamic if and only if $G_{\Lambda, \Sigma}$ is Period Piece$(\Lambda, \Sigma)$-dynamic.

\begin{claim}
    If $G_{\Lambda, \Sigma}$ is Period Piece$(\Lambda, \Sigma)$-dynamic then $G$ is UFP-dynamic.
\end{claim}

\begin{claimproof}
Suppose that $G_{\Lambda, \Sigma}$ is Period Piece$(\Lambda, \Sigma)$-dynamic. Let $f \in F( G_{\Lambda, \Sigma} )$ with $\Per_\Lambda( f ) = \tau$ for some $\tau \in \Sigma$. Let $f^i$ be the restriction of $f$ on $G_i$, and let $h'$ be the restriction of $f$ on $H_\sigma$. Hence $f=(f^1, \dots, f^n, h')$. Suppose, for a contradiction, that each $f^i$ does not have a unique fixed point. By Lemma~\ref{lem:power}, every limit cycle of $f^i$ has length a power of $2$. Let $\ell$ with $n_\ell( \tau ) \geq 1$; thus $\ell \in \Lambda$. Since $\Per_\Lambda( f ) = \tau$, we have $n_\ell(f) = n_\ell(\tau) \geq 1$. But, by Lemma~\ref{lem:two-copies}, we have 
\[
    n_\ell( f ) \ge 2^{n-1} > k \ge n_\ell( \tau ),
\]
a contradiction. Thus $f^i$ has a unique fixed point for some $1\leq i\leq n$. So $G'_i$ is UFP-dynamic and thus $G'$ is UFP-dynamic. 
\end{claimproof}

\begin{claim}
    If 
    $G$ is UFP-dynamic
    then 
    $G_{\Lambda, \Sigma}$ is Period Piece$(\Lambda, \Sigma)$-dynamic.
\end{claim}
\begin{claimproof}
Suppose that $G'$ is UFP-dynamic. For $1\leq i\leq n$, let $f^i\in F(G'_i)$ with a unique fixed point $x^i$. Let $f' = (f^1, \dots, f^n)$ and $f=(f',h)$, where $h$ is the Boolean network, already constructed, with $\Per( h ) = \sigma$ for $\sigma \in \Sigma$ chosen above. Then $f\in F( G_{\Lambda, \Sigma} )$. By Lemma~\ref{lem:long_cycle}, for all $1\leq i\leq n$, all the periodic configurations of $f^i$ distinct from $x^i$ have period at least $2^{N+1}$. We deduce that $x=(x^1,\dots,x^n)$ is the unique fixed point of $f'$ and that all the periodic configurations of $f'$ distinct from $x$ have period at least $2^{N+1}$: there is no $2\leq p<2^{N+1}$ with $n_p(f')\geq 1$. Let $\ell \in \Lambda$; we prove that $n_\ell( f ) = n_\ell( \sigma )$. Since $\Per( h ) = \sigma$, we have $n_\ell(h) = n_\ell(\sigma) \geq 1$. Since $\ell \leq k < 2^{N+1}$, there is no integer $2\leq p\leq\ell$ with $n_p(f')\geq 1$, and we deduce from Lemma~\ref{lem:unique-copy} that 
\[
    n_\ell(f) = n_\ell(f',h) = n_\ell(h) = n_\ell(\sigma).
\]
So $\Per_\Lambda( f ) = \sigma$ and thus $G_{\Lambda, \Sigma}$ is Period Piece$(\Lambda, \Sigma)$-dynamic. 
\end{claimproof}
The theorem follows.
\end{proof}
\end{toappendix}

\bibliography{ref}

@article{PY86-projection,
title = {A note on succinct representations of graphs},
journal = {Information and Control},
volume = {71},
number = {3},
pages = {181-185},
year = {1986},
issn = {0019-9958},
doi = {https://doi.org/10.1016/S0019-9958(86)80009-2},
url = {https://www.sciencedirect.com/science/article/pii/S0019995886800092},
author = {Christos H. Papadimitriou and Mihalis Yannakakis},
}

@book{DF13,
  title={Fundamentals of parameterized complexity},
  author={Downey, Rodney G. and Fellows, Michael R.},
  year={2013},
  publisher={Springer}
}

@article{GJ11,
    author = {Anna G\'al and Jing-Tang Jang},
    title = {The size and depth of layered {B}oolean circuits},
    journal = {Information Processing Letters},
    year = {2011},
    volume = {111},
    pages = {213-217}
}

@INPROCEEDINGS{GI09-tiling-NEXP,
  author={Gottesman, Daniel and Irani, Sandy},
  booktitle={2009 50th Annual IEEE Symposium on Foundations of Computer Science}, 
  title={The Quantum and Classical Complexity of Translationally Invariant Tiling and Hamiltonian Problems}, 
  year={2009},
  volume={},
  number={},
  pages={95-104},
  doi={10.1109/FOCS.2009.22}}

@article{BPR26,
    author = {Florian Bridoux and Aymeric Picard Marchetto and Adrien Richard},
    title = {Interaction graphs of isomorphic automata networks {II}: Universal dynamics},
    journal = {Journal of Computer and System Sciences},
    year = {2026},
    volume = {156},
    number = {103717}
}

@article{Grabisch_2013,
	doi = {10.1016/j.mathsocsci.2013.07.003},
	year = 2013,
	publisher = {Elsevier {BV}},
	volume = {66},
	number = {3},
	pages = {316--330},
	author = {Michel Grabisch and Agnieszka Rusinowska},
	title = {A model of influence based on aggregation functions},
	journal = {Mathematical Social Sciences}
}

@article{Montagud22,
	doi = {10.7554/elife.72626},
	year = 2022,
	publisher = {{eLife} Sciences Publications, Ltd},
	volume = {11},
	author = {Arnau Montagud and Jonas B{\'{e}}al and Luis Tobalina and Pauline Traynard and Vigneshwari Subramanian and Bence Szalai and R{\'{o}}bert Alföldi and L{\'{a}}szl{\'{o}} Pusk{\'{a}}s and Alfonso Valencia and Emmanuel Barillot and Julio Saez-Rodriguez and Laurence Calzone},
	title = {Patient-specific {B}oolean models of signalling networks guide personalised treatments},
	journal = {{eLife}}
}

@book{Rozum_2024, title={Boolean Networks as Predictive Models of Emergent Biological Behaviors}, ISBN={9781009292962}, DOI={10.1017/9781009292955}, publisher={Cambridge University Press}, author={Rozum, Jordan C. and Campbell, Colin and Newby, Eli and Nasrollahi, Fatemeh Sadat Fatemi and Albert, Réka}, year={2024}  }

@article{Gaucherel_2017, title={Understand ecosystem regime shifts by modelling ecosystem development using {B}oolean networks}, volume={31}, ISSN={1476-945X}, DOI={10.1016/j.ecocom.2017.06.001}, journal={Ecological Complexity}, publisher={Elsevier BV}, author={Gaucherel, C. and Théro, H. and Puiseux, A. and Bonhomme, V.}, year={2017}, pages={104–114} }

@article{Poindron_2021, title={A general model of binary opinions updating}, volume={109}, ISSN={0165-4896}, DOI={10.1016/j.mathsocsci.2020.10.004}, journal={Mathematical Social Sciences}, publisher={Elsevier BV}, author={Poindron, Alexis}, year={2021}, pages={52–76} }

@book{P94,
  author    = {Papadimitriou, C. H.},
  title     = {Computational complexity},
  publisher = {Addison-Wesley},
  year      = {1994},
  isbn      = {978-0-201-53082-7},
  bibsource = {dblp computer science bibliography, https://dblp.org}
}

@article{ABGGPRT26,
	author = {Aracena, Julio and Bridoux, Florian and Gadouleau, Maximilien and Guillon, Pierre and Perrot, K{\'e}vin and Richard, Adrien and Theyssier, Guillaume},
	date = {2026/01/21},
	doi = {10.1007/s11047-026-10064-x},
	id = {Aracena2026},
	isbn = {1572-9796},
	journal = {Natural Computing},
	number = {1},
	pages = {6},
	title = {On the dynamics of bounded-degree automata networks},
	url = {https://doi.org/10.1007/s11047-026-10064-x},
	volume = {25},
	year = {2026}}

@article{DFMR23,
  title={An algorithmic pipeline for solving equations over discrete dynamical systems modelling hypothesis on real phenomena},
  author={Dennunzio, Alberto and Formenti, Enrico and Margara, Luciano and Riva, Sara},
  journal={Journal of Computational Science},
  volume={66},
  pages={101932},
  year={2023},
  publisher={Elsevier}
}

@article{bppmr2023,
  title = {{Interaction graphs of isomorphic automata networks I: Complete digraph and minimum in-degree}},
  author = {Bridoux, Florian and Perrot, K{\'e}vin and Picard Marchetto, Aymeric and Richard, Adrien},
  month = dec,
  year = {2023},
  volume = {138},
  pages = {103458},
  journal = {Journal of Computer and System Sciences},
  doi = {10.1016/j.jcss.2023.05.003},
  eprint = {2301.01958}
}

@article{BDPR22,
  title={Complexity of fixed point counting problems in {B}oolean Networks},
  author={Bridoux, Florian and Durbec, Am{\'e}lia and Perrot, K{\'e}vin and Richard, Adrien},
  journal={Journal of Computer and System Sciences},
  year={2022},
  publisher={Elsevier}
}

@article{G20,
  title={On the influence of the interaction graph on a finite dynamical system},
  author={Gadouleau, Maximilien},
  journal={Natural Computing},
  volume={19},
  number={1},
  pages={15--28},
  year={2020},
  publisher={Springer}
}

@article{McC00,
  title={Even dicycles},
  author={McCuaig, William},
  journal={Journal of Graph Theory},
  volume={35},
  number={1},
  pages={46--68},
  year={2000},
  publisher={Wiley Online Library}
}

@article{GR16,
  title={Simple dynamics on graphs},
  author={Gadouleau, Maximilien and Richard, Adrien},
  journal={Theoretical Computer Science},
  volume={628},
  pages={62-77},
  year={2016},
  publisher={Elsevier}
}

@article{R07,
title={Information flows, graphs and their guessing numbers},
  author={Riis, S.},
  journal={The Electronic Journal of Combinatorics},
  volume={14},
  pages={\#R44},
  year={2007},
}

@article{RST99,
  title={Permanents, Pfaffian orientations, and even directed circuits},
  author={Robertson, N. and Seymour, P. and Thomas, R.},
  journal={Annals of Mathematics},
  volume={150},
  number={3},
  pages={929-975},
  year={1999},
  publisher={PRINCETON UNIVERSITY}
}

@book{BG08,
 author = {Bang-Jensen, J. and Gutin, G.~Z.},
 title = {Digraphs: Theory, Algorithms and Applications},
 year = {2008},
 edition = {2nd},
 publisher = {Springer},
}

@ARTICLE{R80,
  author = {Robert, F.},
  title = {Iterations sur des ensembles finis et automates cellulaires contractants},
  journal = {Linear Algebra and its Applications},
  year = {1980},
  volume = {29},
  pages = {393-412}
}

@BOOK{R86,
  title = {Discrete iterations: a metric study},
  publisher = {Springer},
  year = {1986},
  author = {Robert, F.},
  volume = {6},
  series = {Series in Computational Mathematics}
}

@ARTICLE{A08,
  author = {Aracena, J.},
  title = {Maximum Number of Fixed Points in Regulatory {B}oolean Networks},
  journal = {Bulletin of Mathematical Biology},
  year = {2008},
  volume = {70},
  pages = {1398-1409},
  number = {5},
}

@BOOK{R95,
  title = {Les syst\`emes dynamiques discrets},
  publisher = {Springer},
  year = {1995},
  author = {Robert, F.},
  volume = {19},
  series = {Math\'ematiques et Applications}
}

@Article{N15,
  title={Quantitative and logic modelling of molecular and gene networks},
  author={Le Nov{\`e}re, Nicolas},
  journal={Nature Reviews Genetics},
  year={2015},
  volume = 	 {16},
  pages = 	 {146-158},
  publisher={Nature Publishing Group}
}

@Article{DNS12,
  author = 	 {Demongeot, J. and Noual, M. and Sen\'e, S.},
  title = 	 {Combinatorics of {B}oolean automata circuits dynamics},
  journal = 	 {Discrete Applied Mathematics},
  year = 	 {2012},
  OPTkey = 	 {},
  volume = 	 {160},
  number = 	 {4-5},
  pages = 	 {398-415},
  OPTmonth = 	 {},
  OPTnote = 	 {},
  OPTannote = 	 {}
}

\end{document}